\documentclass[twoside,leqno,twocolumn]{article}
\usepackage{ltexpprt}

\usepackage{geometry}

\usepackage{amsmath}
\usepackage{amssymb}
\usepackage{graphicx}
\usepackage{algorithm}
\usepackage{algorithmicx}
\usepackage{algpseudocode}
\algdef{SE}[DOWHILE]{Do}{doWhile}{\algorithmicdo}[1]{\algorithmicwhile\ #1}%
\usepackage{numprint}
\usepackage{xspace}
\usepackage{url}
\usepackage{color}                                                              
\usepackage{colortbl}
\usepackage{authblk}
\usepackage{subfig}
\usepackage{blindtext}
\usepackage{placeins}

\newcommand{\ie}{i.\,e.,\xspace}
\newcommand{\eg}{e.\,g.,\xspace}
\newcommand{\vs}{vs.\ }
\newcommand{\etal}{et al.\xspace}
\newcommand{\wwlog}{W.\,l.\,o.\,g.\ }
\newcommand{\wrt}{w.\,r.\,t.\xspace}
\newcommand{\bigO}{\mathcal{O}}
\newcommand{\argmax}{\operatorname{argmax}}

\newcommand{\IMP}{\operatorname{IMP}}

\newcommand{\peak}{\operatorname{peak}}
\newcommand{\greed}{\textsc{greedy}\xspace}
\newcommand{\HH}{\operatorname{H}}
\newcommand{\GG}{\operatorname{G}}

\newcommand{\dejavu}{\textsc{dejaVu}}

\newcommand{\web}{\textsc{web-Google}\xspace}
\newcommand{\RR}{\mathbb{R}}

\newcommand{\cone}{\textsc{$C_{\vee}$}\xspace}
\newcommand{\ctwo}{\textsc{$C_{\mathcal{P}}$}\xspace}
\newcommand{\cthree}{\textsc{$C_{\wedge}$}\xspace}
\newcommand{\cfour}{\textsc{$C_{m}$}\xspace}
\newcommand{\phistar}{\phi^*}

\newtheorem{notation}{Notation}[section]
\newtheorem{definition}{Definition}[section]
\newtheorem{prop}{Proposition}[section]

\graphicspath{.,./figures/}

\usepackage{ifthen}
\newboolean{double_blind}
\setboolean{double_blind}{false}

\title{\Large Many-to-many Correspondences between Partitions:\\Introducing a Cut-based Approach\thanks{This work is partially supported by DFG grant FINCA (ME-3619/3-1) within the SPP 1736 Algorithms for Big Data.}}

\date{}

\ifthenelse{\boolean{double_blind}}
{
\author{} 
}
{
\author{Roland Glantz\thanks{Faculty of Informatics, Karlsruhe Institute of
  Technology (KIT), Germany, \texttt{rolandglantz@gmail.com}} ~and Henning Meyerhenke\thanks{Institute of Computer Science, University of Cologne, Germany, \texttt{h.meyerhenke@uni-koeln.de}}
}
}

\begin{document}

\maketitle

\begin{abstract}
\small 
Let $\mathcal{P}$ and $\mathcal{P}'$ be finite partitions of the
set $V$. Finding good correspondences between the parts of $\mathcal{P}$
and those of $\mathcal{P}'$ is helpful in classification,
pattern recognition, and network analysis. 
Unlike common similarity measures for partitions that yield only a single value,
we provide specifics on how $\mathcal{P}$ and $\mathcal{P'}$
correspond to each other. 

To this end, we first define natural collections of best
correspondences under three constraints \cone, \ctwo, and \cthree.
In case of \cone, the best correspondences form a
minimum cut basis of a certain bipartite graph, whereas the other two
lead to minimum cut bases of $\mathcal{P}$
\wrt $\mathcal{P}'$. 
We also introduce a constraint, \cfour, which tightens \cthree; 
both are useful for finding consensus partitions.
We then develop branch-and-bound algorithms for finding minimum $P_s$-$P_t$
cuts of $\mathcal{P}$ and thus $\vert \mathcal{P} \vert -1$ best
correspondences under \ctwo, \cthree, and \cfour, respectively.  

In a case study, we use the correspondences to gain insight into a
community detection algorithm. The results suggest, among others, that only very minor losses in the
quality of the correspondences occur if the branch-and-bound
algorithm is restricted to its greedy core. Thus, even for graphs with
more than half a million nodes and hundreds of communities, we can
find hundreds of best or almost best correspondences in less than a minute.\\[0.25ex]
\textbf{Keywords:} Many-to-many correspondences, similarities of partitions,
minimum cut basis, (graph) clustering
\end{abstract}

%
%
\section{Introduction}
\label{sec:intro}
Objective and quantitative methods to help humans with the task of grouping
objects in a meaningful way are the subject of cluster
analysis~\cite{Everitt2011a}. 
We con\-si\-der the case in which
the parts are non-overlapping and form a partition of data points into
parts/clusters/groups/regions/communities.
Even small changes in the data can provoke a clustering algorithm to
split or merge clusters and thus produce different local levels of
detail -- as an example, imagine a (dynamic) clustering algorithm working on data
changing over time. Comparisons of partitions resulting in a single number
expressing total (dis)similarity do not provide specifics
on how the clusters have split or merged, and a comparison restricted
to one-to-one correspondences may be in\-sufficient.

This paper is about a new approach for comparing two partitions
$\mathcal{P} = \{P_1, \dots, P_{\vert \mathcal{P} \vert}\}$ and
$\mathcal{P}' = \{P'_1, \dots, P'_{\vert \mathcal{P}' \vert}\}$ of the
same set.\footnote{If $\mathcal{P}$ and $\mathcal{P}'$ are partitions of
  sets that are different but have a large intersection, $W$, one can
  turn $\mathcal{P}$ and $\mathcal{P}'$ into the two related
  partitions $\{P_1 \cap W, \dots, P_{\vert \mathcal{P} \vert} \cap
  W\}$ and $\{P'_1 \cap W, \dots, P'_{\vert \mathcal{P}' \vert} \cap
  W\}$ of (the same set) $W$. If $W$ is large enough, the
  correspondences between the two new partitions will still reveal
  specifics on similarities between $\mathcal{P}$ and $\mathcal{P}'$.}
In Section~\ref{sec:app:soa} we describe some related work and specify properties that
our approach shares with standard similarity measures for
partitions~\cite{Meila2007a,Wagner2007a}.
The crucial difference is that we provide specifics on how $\mathcal{P}$ and
$\mathcal{P'}$ \emph{correspond} to each other, as opposed to just a single
number.
More specifically, a good many-to-many correspondence, short:
\emph{correspondence}, is a pair $(\mathcal{S}, \mathcal{S}')$ with
$\mathcal{S} \subseteq \mathcal{P}$, $\mathcal{S}' \subseteq
\mathcal{P}'$ and a low value of
\begin{equation}
\label{eq:firstMin1}
\phi(\mathcal{S}, \mathcal{S}') := \vert U_{\mathcal{S}} \triangle
U_{\mathcal{S}'} \vert,
\end{equation}

\noindent where $U_{\mathcal{S}}$ denotes the union of all sets in
$\mathcal{S}$, $\triangle$ denotes the symmetric difference, and
$\vert \cdot \vert$ denotes cardinality [total weight] if the elements
of $V$ are unweighted [weighted]. Thus, minimizing $\phi(\cdot,
\cdot)$ means finding similarities between $\mathcal{P}$ and
$\mathcal{P}'$ modulo unions of parts. 

\begin{figure}[tb]
\begin{center}
\includegraphics[height=2.3cm]{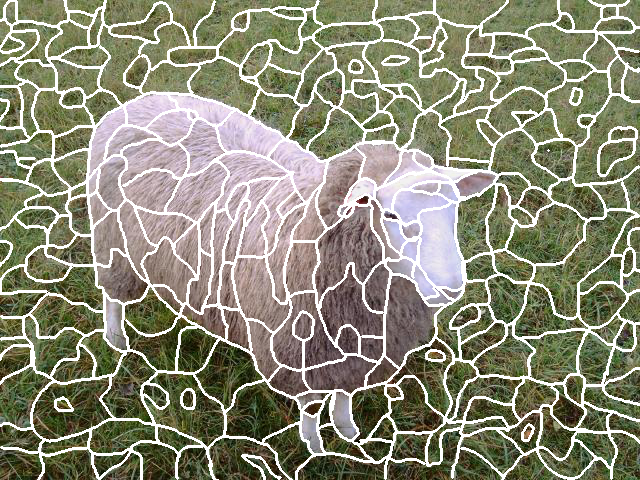} \quad
\includegraphics[height=2.3cm]{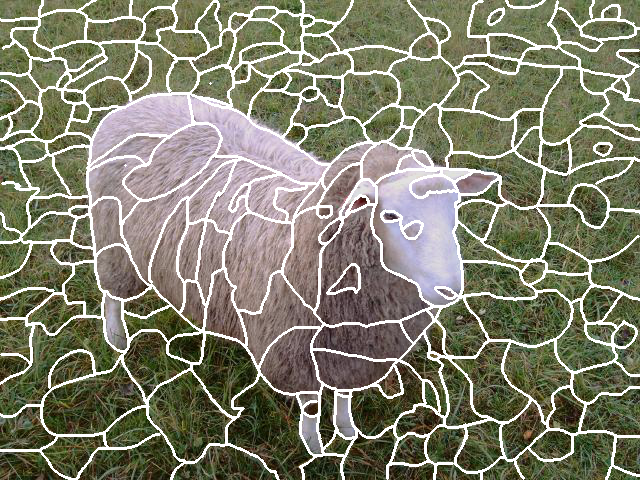}\\
~\\
\includegraphics[height=2.3cm]{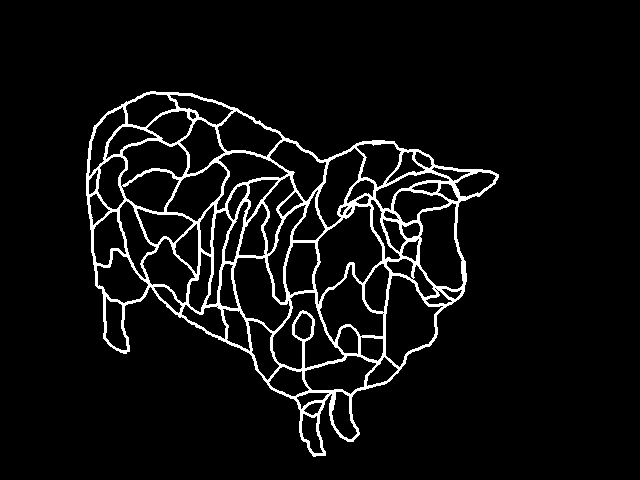} \quad
\includegraphics[height=2.3cm]{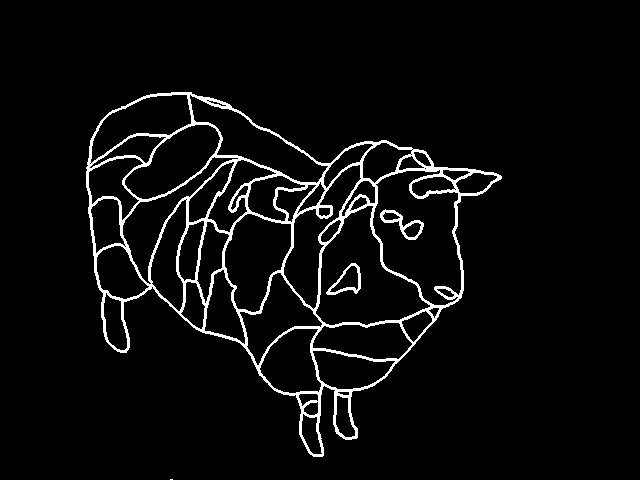}\\
\caption{Top: Two segmentations (by hand) of the same image. Poor
  match between individual regions left and right. Bottom: A good
  correspondence $(\mathcal{S}, \mathcal{S}')$, where $\mathcal{S}$
  and $\mathcal{S}'$ are the regions making up the sheep left and
  right, respectively.}
\label{fig:sheep}
\end{center}
\end{figure}
Among others, correspondences between partitions may be used to describe changes of
ground truth, discrepancies between a model and ground truth, or to
compare different solutions from (variations of) a (possibly
non-deterministic) algorithm. 

To illustrate correspondences further, we turn to applications in which
$\mathcal{P}$ and $\mathcal{P}'$ are segmentations, \ie partitions of
a set of pixels/voxels into regions.
We assume for simplicity that $\mathcal{P}$ and $\mathcal{P}'$ are based on the same
image. Ideally, a region
corresponds to a real-world object; see Figure~\ref{fig:sheep} for an example
of correspondences between different segmentations.
Finding such regions is hindered by noise, under-segmentation, over-segmentation or
occlusion. More scenarios motivating a comparison of $\mathcal{P}$
and $\mathcal{P}'$ using correspondences 
are described in Appendix~\ref{subsec:app:image analysis}.

\vspace{0.5ex}
\paragraph*{Contributions and outline.}
\label{subsec:contrib}
In Section~\ref{sec:problem-statement}, we define the problem and investigate
the connection between correspondences and cuts. We go on by introducing four constraints
(\cone, \ctwo, \cthree, and \cfour) on correspondences, ordered from weak to strong.

Our main objective is to develop methods for finding good correspondences between two partitions of the
same set \wrt all four constraints. 
(Due to space constraints, we focus on \ctwo-correspondences.)
This includes (i) an analytic objective function for finding optimal non-trivial
\ctwo-correspondences and a characterization of the problem in terms
of symmetric submodular minimization (see
Section~\ref{subsec:C2_submod}), (ii) a description of a natural
collection of $\vert \mathcal{P} \vert - 1$ good \ctwo-correspondences
(see Section~\ref{subsec:minST}) and (iii) asymptotic time
complexities for finding natural collections of \ctwo-correspondences
(see Section~\ref{subsec:goodC2_running_times}).


To \emph{compute} good correspondences between two partitions
  $\mathcal{P}$ and $\mathcal{P}'$ in practice, we develop branch-and-bound
  algorithms for finding minimum $P_s$-$P_t$ cuts of $\mathcal{P}$
  under the constraints \ctwo, \cthree and \cfour, respectively, see
  Section~\ref{sec:BB}. The algorithms are built around a greedy
  algorithm each, and the restriction to these greedy cores provides
  an alternative for calculating not always optimal but
  typically good \ctwo-, \cthree- and \cfour-correspondences quickly.

In Section~\ref{sec:real-sim} we use one of many possible
applications to evaluate the correspondence concept and our algorithms
for computing them. 
We investigate the effect that (i) a refinement option and
(ii) non-determinism has on the output of a community
detection (= graph clustering) algorithm.  It turns out that these two effects can indeed
be characterized in terms of correspondences: 
refinement does not change the general cluster assignment
significantly, whereas non-determinism leads to more
drastic changes.  Also, from an algorithmic point of view, only
minor losses in the solution quality 
are observed if the branch-and-bound algorithm is restricted to its greedy
  core. Thus, even for graphs with millions of edges 
  and hundreds of communities, we can find hundreds of best or
almost best correspondences in less than a minute.

\section{Correspondences, cuts, and constraints}
\label{sec:problem-statement}
This section lays the notational ground for 
computing good meaningful correspondences.

\subsection{Correspondences, cuts and optimal partners.}
\label{subsec:cuts}
The element $\mathcal{S}$ of a correspondence $(\mathcal{S}, \mathcal{S}')$
with $\mathcal{S} \notin \{\emptyset, \mathcal{P}\}$ gives rise to a cut $(\mathcal{S},
\mathcal{P} \setminus \mathcal{S})$ of $\mathcal{P}$. We measure the
size (weight) of such a cut by
\begin{equation}
\label{eq:phi}
\phi_{\mathcal{P}'}(\mathcal{S}) := \min_{\mathcal{S}' \subseteq
  \mathcal{P}'}\phi(\mathcal{S}, \mathcal{S}').
\end{equation}

\noindent Given $\mathcal{S} \subseteq \mathcal{P}$, \emph{one} way 
to minimize $\phi(\mathcal{S}, \mathcal{S}')$ is to let $\mathcal{S}'$ be
\begin{equation}
\label{eq:partner}
\mathcal{S_{\downarrow}}' := \{P' \in \mathcal{P}' : \vert U_{\mathcal{S}}
  \cap P' \vert > \frac{\vert P'\vert}{2}\}
\end{equation}

\noindent We call \emph{any} $\mathcal{S}' \in \mathcal{P}'$ with
$\phi(\mathcal{S}, \mathcal{S}') = \phi(\mathcal{S},
\mathcal{S_{\downarrow}}')$ an \emph{optimal partner} of
$\mathcal{S}$. In contrast to $S_{\downarrow}'$, an optimal
  partner of $\mathcal{S}$ may contain $P'$ with $\vert
U_{\mathcal{S}} \cap P' \vert = \vert P'\vert / 2$. A small cut
$(\mathcal{S}, \mathcal{P} \setminus \mathcal{S})$ gives rise to a
good correspondence $(\mathcal{S}, \mathcal{S}')$, where
$\mathcal{S}'$ is an optimal partner of $\mathcal{S}$. In this paper,
we frequently switch between correspondences and cuts.

\subsection{Examples of cuts and correspondences.}
\label{subsec:app:examples-cuts-corres}
\begin{figure}[b]
\begin{center}
\includegraphics[width=0.8\columnwidth]{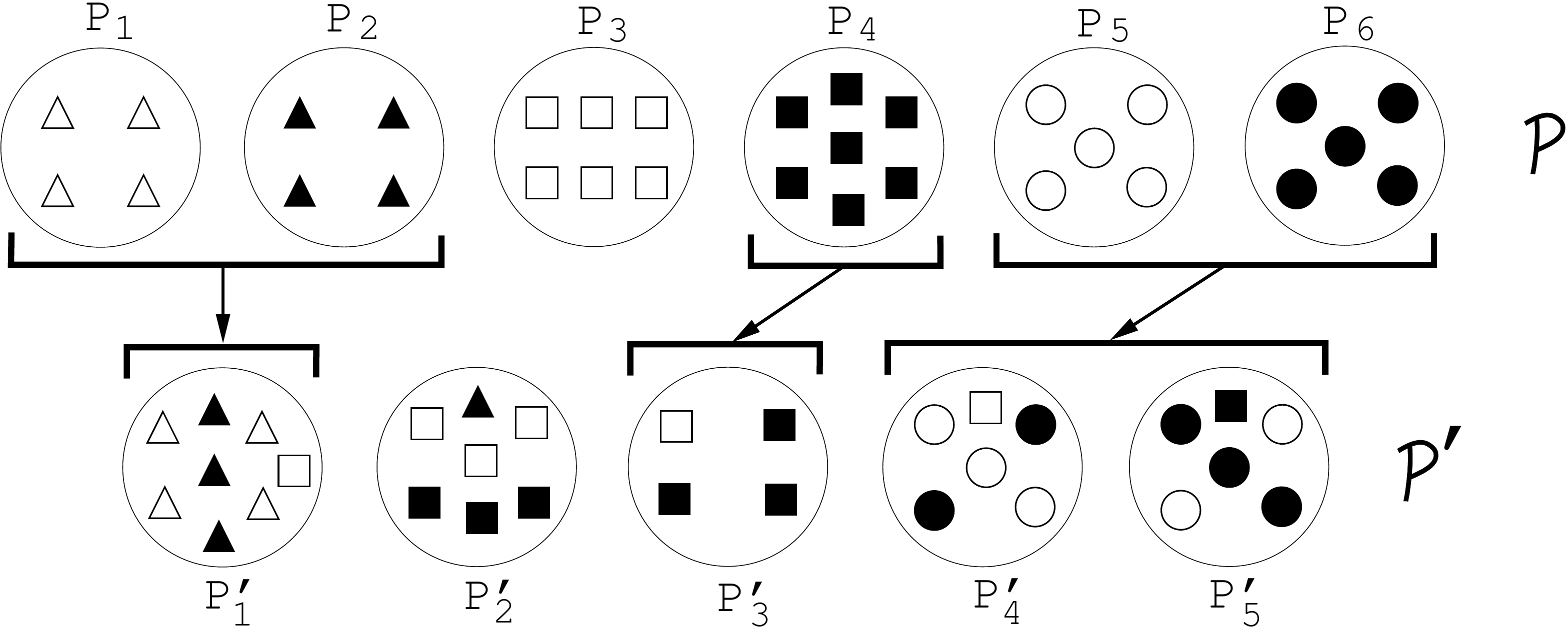}
\caption{Upper [lower] row of six [five] disks depicts partition
  $\mathcal{P}$ [$\mathcal{P}'$]. Upper [lower] brackets indicate
  subsets of $\mathcal{P}$ [$\mathcal{P}'$]; arrows indicate
  some good correspondences between subsets of $\mathcal{P}$ and 
  $\mathcal{P}'$ (see text).}
\label{fig:example}
\end{center}
\end{figure}
Figure~\ref{fig:example} depicts partitions $\mathcal{P}$,
$\mathcal{P}'$ of a set $V$ with 31 elements. The elements of $V$ are
represented by symbols indicating membership to the parts of
$\mathcal{P}$. The two subsets $\mathcal{S}$ of $\mathcal{P}$ giving
rise to the smallest cuts $(\mathcal{S}, \mathcal{P} \setminus
\mathcal{S})$ (size is 2) are the sets $\{P_1, P_2\}$ and $\{P_5,
P_6\}$. The optimal partners of these subsets are the subsets
$\{P'_1\}$ and $\{P'_4, P'_5\}$ of $\mathcal{P}'$, giving rise to the
correspondences $(\{P_1, P_2\}, \{P'_1\})$ and $(\{P_5, P_6\}, \{P'_4,
P'_5\})$, respectively. The optimal partner of $\{P_4\}$ is
$\{P'_3\}$.
If we reverse the roles of $\mathcal{P}$ and $\mathcal{P}'$ (see
Figure~\ref{fig:revExample} in Appendix~\ref{sub:exmp-cuts-corres}), the counterparts of the best
correspondences from before are the new best correspondences, \eg
$(\{P'_1\}, \{P_1, P_2\})$, and $(\{P'_4, P'_5\}, \{P_5, P_6\})$. The
counterpart of $(\{P_4\},\{P'_3\})$, however, is
gone. Indeed, the optimal partner of $\mathcal{S}' = \{P'_3\}$ is not
$\{P_4\}$ but $\emptyset$.
Thus, one must be aware that
swapping the roles of $\mathcal{P}$ and $\mathcal{P}'$ cannot always
be compensated by swapping $\mathcal{S}$ and $\mathcal{S}'$ in a
correspondence.

\subsection{Constraints on correspondences.}
\label{subsec:constraints}
We now define four constraints on correspondences
$(\mathcal{S}, \mathcal{S}')$, ordered from weak to strong, and
suggest cases in which they can be used.

Our first and weakest constraint, called \cone, just excludes correspondences that are trivial or
very bad:
\cone: $\mathcal{S} \notin \{\emptyset, \mathcal{P}\} \vee
\mathcal{S}' \notin \{\emptyset, \mathcal{P}'\}$. For more on
  \cone-correspondences see Appendix~\ref{sec:app:goodC1}.
A more specific constraint that makes sense is \ctwo: $\mathcal{S}
\notin \{\emptyset, \mathcal{P}\}$, \ie that $(\mathcal{S},
\mathcal{P} \setminus \mathcal{S})$ is a cut of
$\mathcal{P}$. \ctwo-correspondences are useful if one wants to
understand the formation of $\mathcal{P}$ in terms of
$\mathcal{P}'$. Exchanging the roles of $\mathcal{P}$ and
$\mathcal{P}'$ yields an analogous asymmetric constraint.

If one wants a correspondence to cut $\mathcal{P}$ \emph{and}
$\mathcal{P}'$, one can require \cthree: $\mathcal{S} \notin
\{\emptyset, \mathcal{P}\} \wedge \mathcal{S}' \notin \{\emptyset,
\mathcal{P}'\}$. In particular, a good \cthree-correspondence gives
rise to two similar cuts $(U_{\mathcal{S}}, U_{\mathcal{P} \setminus
  \mathcal{S}})$ and $(U_{\mathcal{S}'}, U_{\mathcal{P}' \setminus
  \mathcal{S}'})$ of $V$. For such a pair of similar cuts one can find
a cut $(U_{\mathcal{S}^*}, U_{\mathcal{P} \setminus \mathcal{S}^*})$
that mediates between $(U_{\mathcal{S}}, U_{\mathcal{P} \setminus
  \mathcal{S}})$ and $(U_{\mathcal{S}'}, U_{\mathcal{P}' \setminus
  \mathcal{S}'})$ in that $\max\{\vert \mathcal{S}^* \triangle
\mathcal{S} \vert, \vert \mathcal{S}^* \triangle \mathcal{S}' \vert\}$
is minimum. The overlay of $k$ such medial cuts then results in a
consensus partition $\mathcal{P}_c$ between $\mathcal{P}$ and
$\mathcal{P}'$ with $k - 1 \leq \vert \mathcal{P}_c \vert \leq 2^k$
(the medial cuts may or may not cross).
\ctwo-correspondences that do not fulfill \cthree, however, can
provide useful information if one wants to detect erratic differences
between $\mathcal{P}$ and $\mathcal{P}'$. As an example, assume that
$\mathcal{P}$ and $\mathcal{P}'$ consist of communities in a
network at times $t$ and $t' > t$, re\-spec\-tively. Moreover, let $P$
be a community in $\mathcal{P}$. If $(\{P\}, \mathcal{P} \setminus
\{P\})$ is in the minimum cut basis of $\mathcal{P}$ and if
$\emptyset$ is an optimal partner of $\{P\}$, this tells us that $P$
has disintegrated over time in a way that cannot be explained by a
good correspondence between $\mathcal{S}$ and $\mathcal{S'}$
(more on \cthree-correspondences in Appendix~\ref{sec:app:C3C4}).

A correspondence $(\mathcal{S}, \emptyset)$, however, is good, \ie
$\phi(\mathcal{S}, \emptyset)$ is low, whenever $U_{\mathcal{S}}$ is
small (note that $\mathcal{S} \neq \emptyset$ and $\mathcal{S}' =
\emptyset$ fulfill $\vert U_{\mathcal{S}} \triangle U_{\mathcal{S}'}
\vert = \vert U_{\mathcal{S}} \vert$). Even a good
\cthree-correspondence $(\mathcal{S}, \mathcal{S}')$ can be awkward,
\eg if $\vert U_{\mathcal{S}} \vert$ and $\vert U_{\mathcal{S}'}
\vert$ are small and $U_{\mathcal{S}} \cap U_{\mathcal{S}'} =
\emptyset$. Indeed, this means that $\vert U_{\mathcal{S}} \triangle
U_{\mathcal{S}'} \vert$ is still small, \ie the correspondence
$(\mathcal{S}, \mathcal{S}')$ is good, while $\mathcal{S}$ and
$\mathcal{S}'$ ``have nothing in common''. The purpose of
Definition~\ref{def:mutual} is to exclude these correspondences, \ie
to ensure that $\mathcal{S}$ and $\mathcal{S}'$ ``have a lot in
common'':
\begin{definition} 
A correspondence
  $(\mathcal{S}, \mathcal{S}')$ is called \emph{mutual} if all of the
  following holds.

\begin{enumerate}
\setlength\itemsep{1pt}
\item $\vert P \cap U_{\mathcal{S}'} \vert \geq \frac{\vert P
  \vert}{2}$ for all $P \in \mathcal{S}$,
\item $\vert P \cap U_{\mathcal{S}'} \vert \leq \frac{\vert P
  \vert}{2}$ for all $P \in \mathcal{P} \setminus \mathcal{S}$,
\item $\vert P' \cap U_{\mathcal{S}} \vert \geq \frac{\vert P'
  \vert}{2}$ for all $P' \in \mathcal{S}'$ and
\item $\vert P' \cap U_{\mathcal{S}} \vert \leq \frac{\vert P'
  \vert}{2}$ for all $P' \in \mathcal{P}' \setminus \mathcal{S}'$.
\end{enumerate}

\label{def:mutual}
\end{definition}

If $(\mathcal{S}, \mathcal{S}')$ with $\mathcal{S} \notin \{\emptyset,
\mathcal{P}\}$ is mutual, then $\mathcal{S}' \notin \{\emptyset,
\mathcal{P}\}$ too. Thus, a new meaningful constraint on
correspondences $(\mathcal{S}, \mathcal{S}')$ that is stronger than
\cthree is: 
\begin{equation*}
\cfour: \mathcal{S} \notin \{\emptyset, \mathcal{P}\} \wedge
\mbox{$(\mathcal{S}, \mathcal{S}')$ is mutual}.
\end{equation*}
For more on \cfour-correspondences see Appendix~\ref{sec:app:C3C4}.
In Appendix~\ref{subsec:app:mutual} we show that an optimal
\ctwo-correspondence or \cthree-correspondence is either mutual or
fulfills $\vert \mathcal{S} \vert \in \{1, \vert
\mathcal{P} \vert - 1\} \vee \vert \mathcal{S}' \vert \in \{1, \vert
\mathcal{P}' \vert - 1\}$.

\section{\ctwo-correspondences}
\label{sec:goodC2}
In this section we 
reduce the problem of finding a good
\ctwo-correspondence $(\mathcal{S}, \mathcal{S}')$ to the problem of
finding a small cut $(\mathcal{S}, \mathcal{P} \setminus \mathcal{S})$
of $\mathcal{P}$.
We then show that $\phi_{\mathcal{P}'}(\cdot)$ is a symmetric
submodular function on $2^{\mathcal{P}}$. Symmetry of
$\phi_{\mathcal{P}'}(\cdot)$ implies that one can find a minimum cut
basis of $\mathcal{P}$ by computing $\vert \mathcal{P} \vert - 1$
minimum $P_s$-$P_t$ cuts of $\mathcal{S}$. 
Finally, we discuss asymptotic
running times for finding good \ctwo-correspondences.

\subsection{\ctwo-correspondences and submodularity.}
\label{subsec:C2_submod}
The constraint \ctwo: $\mathcal{S} \notin \{\emptyset, \mathcal{P}\}$
on a correspondence $(\mathcal{S}, \mathcal{S}')$ does not constrain
$\mathcal{S}'$. The search for a ``good'' \ctwo-correspondence thus
basically amounts to finding $\emptyset \neq \mathcal{S} \subsetneq
\mathcal{P}$ such that $\phi_{\mathcal{P}'}(\mathcal{S}) :=
\min_{\mathcal{S}' \in \mathcal{P}'}\phi(\mathcal{S}, \mathcal{S}')$
(Eq.~(\ref{eq:phi})) is ``low''. Once we have $\mathcal{S}$,
we can find an optimal partner $\mathcal{S}'$ of $\mathcal{S}$ via
Eq.~(\ref{eq:partner}). The problem with this approach is that,
in its present form, $\phi_{\mathcal{P}'}(\mathcal{S})$ depends on
$\mathcal{S}'$. Proposition~\ref{prop:eq:firstMin1}  provides an
analytic expression for $\phi_{\mathcal{P}'}(\cdot)$ that does not
contain $\mathcal{S}'$:
\begin{prop}[Proof in Appendix~\ref{proof:eq:firstMin1}]
\label{prop:eq:firstMin1}
\begin{align}
&\phi_{\mathcal{P}'}(\mathcal{S}) =\sum_{P' \in
    \mathcal{P}'} \vert P' \vert \peak(\frac{\vert U_{\mathcal{S}}
    \cap P' \vert}{\vert P' \vert}), \quad \label{eq:Frac}
  \mbox{where} \\ &\peak(x) := \begin{cases} x, &
    \mbox{if } x \leq 1/2\\ 1 - x, & \mbox{if } x >
    1/2 \end{cases} \label{eq:peak}
\end{align}
\end{prop}


The function $\peak(\cdot)$ in Eq.~(\ref{eq:peak}), called
(classification) error in~\cite{Tan2005a}, is an example of a
\emph{generator} as defined in~\cite{Simovici2002a}: a function $f :
     [0,1] \mapsto \RR$ is a generator if it is concave and $f(0) =
     f(1) = 0$ (hence $f(\cdot)$ is also subadditive). In addition,
     $\peak(\cdot)$ is symmetric, \ie $\peak(p) = \peak(1-p)$ for all
     $p \in [0,1]$. Other examples of symmetric generators are the
     binary entropy function
     $H(\cdot)$~\cite{MacKay2003a,Tan2005a,Simovici2002a} and the Gini
     impurity measure $G(\cdot)$~\cite{Tan2005a,Simovici2002a}. We
     could have chosen $\HH(\cdot)$, $\GG(\cdot)$ or any other
     nontrivial symmetric generator (computable in constant time)
     instead of $\peak(\cdot)$. The minimization of
     Eq.~(\ref{eq:Frac}) would then have the same asymptotic time
     complexity (see Section~\ref{subsec:goodC2_running_times}). 

\begin{definition}[Symmetric, (sub)modular]
\label{def:submodular}
Let $\mathcal{P}$ be a set. A function $\Pi : 2^{\mathcal{P}} \mapsto
\RR $ is called \emph{symmetric} if $\Pi(\mathcal{S}) =
\Pi(\mathcal{P} \setminus \mathcal{S})$ for all $\mathcal{S} \subseteq
\mathcal{P}$. Furthermore, $\Pi(\cdot)$ is called \emph{submodular} if
$\Pi(\mathcal{S}_1 \cup \mathcal{S}_2) \leq \Pi(\mathcal{S}_1) +
\Pi(\mathcal{S}_2) - \Pi(\mathcal{S}_1 \cap \mathcal{S}_2)$ for all
$\mathcal{S}_1, \mathcal{S}_2 \subseteq \mathcal{P}$. If $\Pi(\cdot)$
fulfills the above with ``$=$'' instead of ``$\leq$'', then
$\Pi(\cdot)$ is called \emph{modular}.
\end{definition}

\begin{prop}[Proof in Appendix~\ref{proof:prop:submodular}]
\label{prop:submodular}
$\phi_{\mathcal{P}'}(\cdot)$ in Eq.~(\ref{eq:Frac}) is symmetric
and submodular.
\end{prop}

\subsection{Minimum $P_s$-$P_t$ cuts.}
\label{subsec:minST}
We are actually interested in a larger set of
\emph{good} correspondences between $\mathcal{P}$ and $\mathcal{P}'$
(rather than a single one)
or, equivalently, in a larger set of small
cuts of $\mathcal{P}$. A natural set of small cuts is formed by
minimum $P_s$-$P_t$ cuts:
\begin{definition} 
\label{def:s-t-cuts}
Let $P_s \neq P_t \in \mathcal{P}$. Any pair $(\mathcal{S}_s,
\mathcal{S}_t)$ with $P_s \in \mathcal{S}_s$, $P_t \in \mathcal{S}_t$
and $\mathcal{S}_t = \mathcal{P} \setminus \mathcal{S}_s$ is called a
\emph{$P_s$-$P_t$ cut} of $\mathcal{P}$. A $P_s$-$P_t$ cut $(\mathcal{S}_s,
\mathcal{S}_t)$ is minimum if $\phi_{\mathcal{P}'}(\mathcal{S}_s)$,
and thus $\phi_{\mathcal{P}'}(\mathcal{S}_t)$, is minimum
\wrt all $P_s$-$P_t$ cuts. 
\end{definition}

Analogous to graphs, there exists a minimum cut basis of $\mathcal{P}$
\wrt $\phi_{\mathcal{P}'}(\cdot)$ made up of $\vert
\mathcal{P} \vert - 1$ minimum $P_s$-$P_t$ cuts. This follows from
$\phi_{\mathcal{P}'}(\cdot)$ being
symmetric~\cite{Cheng1992a}. Moreover, the cuts in the minimum basis
are non-crossing (two cuts are non-crossing if their cut sides are
pairwise nested or disjoint~\cite{HartmannW12cut}). This is a consequence of
$\phi_{\mathcal{P}'}(\cdot)$ being
submodular~\cite{Gomory1961a,Queyranne98a}. A minimum
basis of cuts of $\mathcal{P}$ can be represented concisely by a Gomory-Hu tree~\cite{Gomory1961a}.



\subsection{Asymptotic time for minimum cut basis of \ctwo-correspondences.}
\label{subsec:goodC2_running_times}
To compute a minimum basis of cuts of $\mathcal{P}$ under the
constraint \ctwo, we have to compute $\vert \mathcal{P} \vert - 1$
minimum $P_s$-$P_t$
cuts~\cite{Gomory1961a,Gusfield90a}. Unfortunately, computing such a
cut in the setting of general symmetric submodular minimization is as
hard as minimizing a general non-symmetric submodular
function~\cite{Queyranne98a}; $\bigO(\vert \mathcal{P}
\vert^7) \log \vert \mathcal{P} \vert$ evaluations of
$\phi_{\mathcal{P}'}(\cdot)$ would be needed to find a single
minimum $P_s$-$P_t$ cut~\cite[Theorem
  4.3]{Iwata2001a}. Fortunately, finding $P_s$-$P_t$ cuts is
  easier in our case (proof in
  Appendix~\ref{proof:prop:run-time-ctwo}):

%
\begin{prop}
\label{prop:running-nontrivial:P}
A minimum cut basis of $\mathcal{P}$ \wrt
$\phi_{\mathcal{P}'}(\cdot)$ can be computed in time
\begin{equation}
\bigO(\vert V \vert + \vert \mathcal{P} \vert^3 \vert \mathcal{P}'
\vert + \vert \mathcal{P} \vert^2 \vert \mathcal{P}'
\vert^2)
\label{firstBigO}
\end{equation}
\vspace{-0.25cm}\centering \text{or in}\vspace{-0.25cm} \\
\begin{equation}
\label{secondBigO}
\begin{cases}
\bigO(\vert V \vert + \vert \mathcal{P} \vert^3 \vert \mathcal{P}' \vert \log (2 + 
\vert \mathcal{P} \vert^2 / \vert \mathcal{P}' \vert)), & \text{if~}
\vert \mathcal{P} \vert \leq \vert \mathcal{P}' \vert\\
\bigO(\vert V \vert + \vert \mathcal{P} \vert^2 \vert \mathcal{P}' \vert^2 \log (2 + 
\vert \mathcal{P}' \vert^2 / \vert \mathcal{P} \vert)), &
\text{otherwise}
\end{cases}
\end{equation}
\end{prop}
%


The new notation in Definition~\ref{def:distribs} below helps
to prove Propositions~\ref{prop:evaluateFrac} and~\ref{prop:complexity}.
\begin{definition}[Distributions $d_{P'}\lbrack \cdot \rbrack$]
\label{def:distribs}
Let $P' \in \mathcal{P}'$. The distribution of $P'$ \wrt $\mathcal{P}$
is the vector $d_{P'}[\cdot]$ of length $\vert \mathcal{P} \vert$
defined by $d_{P'}[i] := \vert P_i \cap P' \vert \mbox{~for~} 1 \leq i
\leq \vert \mathcal{P} \vert$.
\end{definition}
The computation of all distributions (necessary to compute $\vert
U_{\mathcal{S}} \cap P' \vert$ in Eqs.~(\ref{eq:partner})
and~(\ref{eq:Frac})), \ie the contingency table~\cite{Wagner2007a}, takes
time $\bigO(\vert V \vert + \vert \mathcal{P} \vert \vert \mathcal{P}'
\vert)$, see Appendix~\ref{subsec:goodC1_running_times}. 
%
The next result, Proposition~\ref{prop:evaluateFrac},
follows directly from the fact that, due to
$\vert U_{\mathcal{S}} \cap P' \vert = \sum_{i: P_i \in \mathcal{S}}
d_{P'}[i]$, the term $\vert U_{\mathcal{S}} \cap P' \vert$ can be
computed in $\bigO(\vert \mathcal{P} \vert)$ for any $P' \in
\mathcal{P}'$. It allows to derive Proposition~\ref{prop:complexity} afterwards.
\begin{prop}
\label{prop:evaluateFrac}
Given all distributions and $\mathcal{S} \subseteq \mathcal{P}$, the
calculation of $\mathcal{S}'$, as defined in
Eq.~(\ref{eq:partner}), and the evaluation of
$\phi_{\mathcal{P}'}(\mathcal{S})$, as defined in
Eq.~(\ref{eq:Frac}), can both be done in time $\bigO(\vert
\mathcal{P} \vert \vert \mathcal{P}' \vert)$.
\end{prop}

\begin{prop}
\label{prop:complexity}
Finding an optimal \ctwo-correspondence takes time $\bigO(\vert V
\vert + \vert \mathcal{P} \vert^4 \vert \mathcal{P}' \vert)$ if one
first minimizes $\mathcal{S}$ in Eq.~(\ref{eq:Frac}) through
general symmetric submodular minimization and then determines the
optimal partner $\mathcal{S}'$ of $\mathcal{S}$ using
Eq.~(\ref{eq:partner}). For a proof see
Appendix~\ref{proof:prop:complexity}.
\end{prop}

Interestingly, the asymptotic time for computing the minimum cut basis
of $\mathcal{P}$ is lower than that for computing just one optimal
\ctwo-correspondence using symmetric submodular minimization.\footnote{The only
way the latter time could be lower than or as low as $\bigO(\vert V
\vert + \vert \mathcal{P} \vert^3 \vert \mathcal{P}' \vert + \vert
\mathcal{P} \vert^2 \vert \mathcal{P}' \vert^2)$ (see
Eq.~(\ref{firstBigO})) would entail $\vert \mathcal{P} \vert^4
\vert \mathcal{P}' \vert \leq \vert \mathcal{P} \vert^2 \vert
\mathcal{P}' \vert^2$. 
Eq.~(\ref{secondBigO}) then yields that the
minimum cut basis can be computed in time $\bigO(\vert V \vert + \vert
\mathcal{P} \vert^3 \vert \mathcal{P}' \vert)$.}

\section{Computing minimum $P_s$-$P_t$ cuts}
\label{sec:BB}
With the intent of improving the running time of the results
in the previous section for practical purposes, we continue with branch-and-bound (B\&B) 
and greedy techniques.

\subsection{Basic branch-and-bound algorithm.}
\label{subsec:BB_basicAlgo}
Let $P_s \neq P_t \in \mathcal{P}$. Our goal is to find a minimum
$P_s$-$P_t$ cut $(\mathcal{S}_s, \mathcal{S}_t)$ of $\mathcal{P}$ \wrt
$\phi_{\mathcal{P}'}(\cdot)$.
The idea be\-hind our algorithm is to first set $\mathcal{S}_s :=
\{P_s\}$, $\mathcal{S}_t := \{P_t\}$ and then let $\mathcal{S}_s$ and
$\mathcal{S}_t$ compete for the remaining parts in $\mathcal{P}$ until
$\mathcal{S}_s \cup \mathcal{S}_t = \mathcal{P}$. To curtail the
exponentially growing number of possibilities that arise when
assigning new parts, \ie parts in $\mathcal{P} \setminus
(\mathcal{S}_s \cup \mathcal{S}_t)$, to either $\mathcal{S}_s$ or
$\mathcal{S}_t$, we need a lower bound $b(\mathcal{S}_s \cup
\mathcal{S}_t)$ on how low $\phi_{\mathcal{P}'}(\mathcal{S})$ can
possibly get for $\mathcal{S}$ with $\mathcal{S}_s \subseteq
\mathcal{S}$ and $\mathcal{S} \cap \mathcal{S}_t =
\emptyset$. Proposition~\ref{prop:bound} below guarantees that the
bound defined next is admissible.

\begin{definition}
\label{def:bound}
Let $\mathcal{S}_s, \mathcal{S}_t \subseteq \mathcal{P}$ with $P_s \in
\mathcal{S}_s$, $P_t \in \mathcal{S}_t$ and $\mathcal{S}_s \cap
\mathcal{S}_t = \emptyset$. We set $b(\mathcal{S}_s, \mathcal{S}_t)
:=\sum_{P' \in \mathcal{P}'} \min\{\vert U_{\mathcal{S}_s} \cap P'
\vert, \vert U_{\mathcal{S}_t} \cap P' \vert\}$.
%
%
\end{definition}

\begin{prop}[Proof in Appendix~\ref{proof:bound}]
\label{prop:bound}
Let $\mathcal{S}_s, \mathcal{S}_t \subseteq \mathcal{P}$ with $P_s \in
\mathcal{S}_s$, $P_t \in \mathcal{S}_t$ and $\mathcal{S}_s \cap
\mathcal{S}_t = \emptyset$. Moreover, let $\mathcal{S} \supseteq
\mathcal{S}_s$, $\mathcal{S} \cap \mathcal{S}_t = \emptyset$. Then,
$b(\mathcal{S}_s, \mathcal{S}_t) \leq
\phi_{\mathcal{P}'}(\mathcal{S})$.
\end{prop}

\noindent We still have to make decisions on (i) the choice of
  the next part $P$ from $\mathcal{P} \setminus (\mathcal{S}_s \cup
  \mathcal{S}_t)$ that we use to extend either $\mathcal{S}_s$ or
  $\mathcal{S}_t$ and (ii) whether we assign $P$ to $\mathcal{S}_s$ or
  $\mathcal{S}_t$.
Our strategy for (ii) is to assign $P$ to $\mathcal{S}_s$ or
$\mathcal{S}_t$ according to the (optimistic) prospect
$b(\mathcal{S}_s, \mathcal{S}_t)$, \ie $P$ is assigned such that the
new value $b(\mathcal{S}_s, \mathcal{S}_t)$ is minimum. We prefer
minimizing $b(\mathcal{S}_s, \mathcal{S}_t)$ over minimizing
$\phi_{\mathcal{P}'}(\mathcal{S}_s)$ and/or
$\phi_{\mathcal{P}'}(\mathcal{S}_t)$ because the latter two numbers
can both be high although the prospect for finding a good cut of
$\mathcal{P}$ is still good. Due to the definition of $b(\cdot,
\cdot)$ and the symmetry of $\phi_{\mathcal{P}'}(\cdot)$, however, the
objectives of minimizing $b(\mathcal{S}_s, \mathcal{S}_t)$,
$\phi_{\mathcal{P}'}(\mathcal{S}_s)$ and
$\phi_{\mathcal{P}'}(\mathcal{S}_t)$ will have converged by the time
when $\mathcal{S}_s$ and $\mathcal{S}_t$ are fully grown, \ie
$\mathcal{S}_s \cup \mathcal{S}_t = \mathcal{P}$.
Our strategy for (i), the choice of $P$, aims at shifting
the backtracking phases of our B\&B algorithm to scenarios
in which $\mathcal{S}_s$ and $\mathcal{S}_t$ are already large and
where chances are that we are close to a new minimum of
$\phi_{\mathcal{P}'}(\cdot)$. To this end, we pick $P$ from
$\mathcal{P} \setminus (\mathcal{S}_s \cup \mathcal{S}_t)$ such that
the alternative between putting $P$ into $\mathcal{S}_s$ or into
$\mathcal{S}_t$ matters the most in terms of $b(\cdot,
\cdot)$. Formally,
\begin{equation}
\small 
\label{eq:nextP}
P = \argmax\limits_{P \in \mathcal{P}
\setminus (\mathcal{S}_s \cup \mathcal{S}_t)} \vert b(\mathcal{S}_s \cup \{P\},
\mathcal{S}_t) - 
b(\mathcal{S}_s, \mathcal{S}_t \cup \{P\}) \vert
\end{equation}

After initializing $\mathcal{S}_s$ and $\mathcal{S}_t$, our
B\&B algorithm calls $\greed(\mathcal{S}_s, \mathcal{S}_t,
\infty)$, which is shown as Algorithm~\ref{algo:BB} in Section~\ref{subsec:greedy}. In later calls of
$\greed(\mathcal{S}_s, \mathcal{S}_t, bestSoFar)$, we always have
$(\mathcal{S}_s \supsetneq \{P_s\} \vee \mathcal{S}_t \supsetneq
\{P_t\}) \wedge \mathcal{S}_s \cap \mathcal{S}_t = \emptyset$, and
$bestSoFar$ amounts to the minimum weight ($\phi_{\mathcal{P}'}$
value) of the $P_s$-$P_t$ cuts found so far (see lines 5-10 of
Algorithm~\ref{algo:BB} in Appendix~\ref{subsec:app:bb}). Crucial questions \emph{after} any call
of $\greed$ are
\begin{itemize}
\item[1)] whether $\greed(\cdot, \cdot, \cdot)$ needs to be invoked
  again and, \emph{if so},
\item[2a)] which of the most recent assignments of parts (to
  $\mathcal{S}_s$ or $\mathcal{S}_t$) should be undone when
  backtracking and 
\item[2b)] which alternative line for searching a minimum $P_s$-$P_t$
  cut is taken after backtracking, \ie what is the input for the next
  $\greed(\cdot, \cdot, \cdot)$ call.
\end{itemize}

The answer to 1) is ``as long as $\mathcal{S}_s \supsetneq \{P_s\}$ or
$\mathcal{S}_t \supsetneq \{P_t\}$''. In other words, we stop when
$(\mathcal{S}_s, \mathcal{S}_t)$ has shrunk to its initialization
$(\{P_s\}, \{P_t\})$ (see lines 3 and 23 of Algorithm~\ref{algo:BB}).
The answer to 2a) and 2b), in turn, as well as more details on our B\&B  algorithm (such as pseudocode),
can be found in Appendix~\ref{subsec:app:bb}.
How to extend the B\&B algorithm from \ctwo to \cthree and \cfour
is described in Appendix~\ref{subsubsec:extensions}.

\subsection{Speeding up the B\&B algorithm.}
\label{subsec:BB_extensions}
%
As before, let $\mathcal{S}_s$ and $\mathcal{S}_t$ consist of the parts
of $\mathcal{P}$ that have already been assigned to the $P_s$-side and
the $P_t$-side of the cut, respectively. To tighten the bound at the
current state of assembling $\mathcal{S}_s$ and $\mathcal{S}_t$, we
take a closer look at parts $P' \in \mathcal{P}'$ that overlap well
with $\mathcal{S}_s$ or $\mathcal{S}_t$ already. Specifically, let $P'
\in \mathcal{P}'$ such that $\vert P' \cap U_{\mathcal{S}_t} \vert
\geq \vert P' \vert / 2$. Then, if no backtracking behind the current
state occurs, a new assignment of some $P \in \mathcal{P}$ to the
$P_s$-side will increase the value of $\phi_{\mathcal{P}'}(\cdot)$ by
at least $I_s(P, P') := \vert P \cap P' \vert$. Exchanging the roles
of $s$ and $t$ may yield alternative increases $I_t(P, P')$ (based on
other $P'$). Thus,

\begin{equation*}
I_s(P) := \sum(I(P, P') \mbox{~:~} P' \mbox{ fulfills } \vert P' \cap U_{\mathcal{S}_t}
\vert \geq \vert P' \vert / 2)
\end{equation*}

\noindent and analogously defined $I_t(P)$ are the increases of
$\phi_{\mathcal{P}'}(\cdot)$ if $P$ is assigned to the $s$-side or
$t$-side, respectively. Hence, summing up the terms $\min\{I_s(P),
I_t(P)\}$ over all $P$ not yet assigned to any side yields a lower
bound on the future increase of the objective function.
%
Apart from improving the bound, a second way to curtail the search is
to interpret current $\mathcal{S}_s$ and $\mathcal{S}_t$ as two
\ctwo-correspondences $(\mathcal{S}_s, \mathcal{S}_s')$ and
$(\mathcal{S}_t, \mathcal{S}_t')$, where $\mathcal{S}_s'$ and
$\mathcal{S}_t'$ are optimal partners of $\mathcal{S}_s$ and
$\mathcal{S}_t$, respectively.

\subsection{Greedy heuristic.}
\label{subsec:greedy}
  Algorithm~\ref{algo:greedy}, $\greed$ $(\mathcal{S}_s, \mathcal{S}_t, bestSoFar)$, is at the heart of our B\&B 
  algorithm. It greedily extends a pair $(\mathcal{S}_s,
  \mathcal{S}_t)$ and terminates prematurely, \ie with $\mathcal{S}_s
  \cup \mathcal{S}_t \neq \mathcal{P}$, if there is no chance to find
  $\mathcal{S}$ with $\phi_{\mathcal{P}'}(\mathcal{S}) < bestSoFar$.
In the first call of $\greed(\mathcal{S}_s, \mathcal{S}_t, bestSoFar)$,
we have $\mathcal{S}_s = \{P_s\}$, $\mathcal{S}_s = \{P_s\}$ and
$bestSoFar = \infty$. In particular, $\greed(\{P_s\}, \{P_t\},
\infty)$ does \emph{not} end prematurely, \ie it delivers a
$P_s$-$P_t$ cut $(\mathcal{S}, \mathcal{P} \setminus \mathcal{S})$.
While $\greed$ does not guarantee optimality, it will be interesting
if its quality is acceptable in practice.

\begin{algorithm}[tb]
  \caption{Algorithm $\greed(\mathcal{S}_s, \mathcal{S}_t, bestSoFar)$
    for extending a pair $(\mathcal{S}_s, \mathcal{S}_t)$ with
    $\mathcal{S}_s \cap \mathcal{S}_t = \emptyset$ towards a pair
    $(\mathcal{S}, \mathcal{P} \setminus \mathcal{S})$ with
    $\mathcal{S} \supseteq \mathcal{S}_s$ and $\mathcal{S} \cap
    \mathcal{S}_t = \emptyset$ as long as there is a chance that
    $\phi_{\mathcal{P}'}(\mathcal{S}) < bestSoFar$}
  \label{algo:greedy}
  \begin{algorithmic}[1]
    \While{$(\mathcal{S}_s \cup \mathcal{S}_t \neq \mathcal{P})
      \wedge b((\mathcal{S}_s, \mathcal{S}_t) < bestSoFar)$}
    \State Find $P \in \mathcal{P} \setminus (\mathcal{S}_s \cup
    \mathcal{S}_t)$ that fulfills Eq.~(\ref{eq:nextP})
    \If{$b(\mathcal{S}_s \cup \{P\}, \mathcal{S}_t) <
      b(\mathcal{S}_s, \mathcal{S}_t \cup \{P\})$}
    \State $\mathcal{S}_s \gets \mathcal{S}_s \cup \{P\}$
    \Else
    \State $\mathcal{S}_t \gets \mathcal{S}_t \cup \{P\}$
    \EndIf
    \EndWhile
  \end{algorithmic}
\end{algorithm}

%
%
\section{Correspondences in community detection}
\label{sec:real-sim}
In the experiments of this section we
not only evaluate the performance of our B\&B and $\greed$ algorithms,
but also gain insight into different variants of the Louvain method, a
community detection algorithm. Community detection is a graph
clustering problem well-known in (social) network
analysis~\cite{Fortunato201075}, resulting in a partition of the
graph's node set. Finding correspondences between communities is
challenging when the nodes come without attributes that could help
with the task (like colored pixels in images), which is the case here.

\subsection{Louvain method (LM) and variants.}
\label{subsec:LM}
LM~\cite{Blondel:2008uq} is a locally greedy, bottom-up multilevel
algorithm. It is very popular for community detection by maximizing
the objective function modularity (this
modularity~\cite{girvan2002community} should not be confused with the
one in Definition~\ref{def:submodular}).  On each hierarchy level, LM
assigns nodes to communities iteratively, while
maximizing modularity greedily.
The communities on each level are contracted into single
nodes, giving rise to the graph on the next level. The solution of
the coarsest graph is then successively expanded to the next finer
level, respectively.
In~\cite{DBLP:journals/tpds/StaudtM16} a shared-memory parallelization
of LM, called PLM, is provided.
PLM is not deterministic since the outcome depends on the order of the threads. We
denote the single-threaded (sequential) version of PLM by SLM. Both
versions have been extended by an optional \emph{refinement} phase:
after each expansion, nodes are again moved for modularity gain. SLM
with refinement is denoted by SLMR.

\subsection{Research questions and approach.}
Let now $\mathcal{P}$ and $\mathcal{P}'$ be from SLM and SLMR,
  respectively. We first want to know whether the transition from
  $\mathcal{P}$ to $\mathcal{P}'$ is best described as (a) communities
  merely exchanging elements with each other, but otherwise remaining
  as they are or (b) involving unions and break-ups of communities. An
  analogous question arises when $\mathcal{P}$ and $\mathcal{P}'$ are
  the partitions returned by different (non-deterministic) runs of
  PLM.

Second, we want to test if the choice between \ctwo- and
   \cthree-correspondences matters when comparing communities.
Third, we want to evaluate the tradeoff between quality and running time 
 for both the B\&B algorithm and the heuristic $\greed$.

To answer the first question, we use the $\vert \mathcal{P} \vert - 1$
best \ctwo-correspondences: if the
communities merely exchange elements with each other, most
\ctwo-correspondences $(\mathcal{C}, \mathcal{C}')$ should be such
that $\vert \mathcal{C} \vert = \vert \mathcal{C}' \vert$. Conversely,
unions and break-ups of communities should result in many instances
with $\vert \mathcal{C} \vert \neq \vert \mathcal{C}' \vert$.



To answer the second and the third question, we compute the
$\vert \mathcal{P} \vert - 1$ best \ctwo- and the $\vert \mathcal{P}
\vert - 1$ best \cthree-correspondences with the
B\&B algorithm and the heuristic $\greed$, respectively. For each of these
  four scenarios, we then aggregate the $\vert \mathcal{P} \vert - 1$
best correspondences by calculating what can be called \emph{total
  dissimilarity}, \ie the sum of the $\phi_{\mathcal{P}'}(\cdot,
\cdot)$ values of the $\vert \mathcal{P} \vert - 1$ best
correspondences, divided by the total number of vertices.
In experiments involving (non-deterministic) PLM, we smooth total dissimilarity by averaging over 10 runs.

We compare the results from the four scenarios as follows. Let
  $d_{\mathcal{P}}$, $d_{\wedge}$, $d^h_{\mathcal{P}}$ and
  $d^h_{\wedge}$ be the total dissimilarity from the scenarios (i)
  \ctwo, B\&B, (ii) \cthree, B\&B, (iii)
  \ctwo, heuristic, and (iv) \cthree, heuristic, respectively. We form
  the ratios $r_1 = d_{\mathcal{P}}/d_{\wedge}$, $r_2 =
  d_{\mathcal{P}}/d^h_{\mathcal{P}}$ and $r_3 =
  d_{\mathcal{P}}/d^h_{\wedge}$. In case of experiments
  involving PLM, different scenarios come with different sets of ten PLM-generated partitions each.
%
%
As input for PLM, SLM, and SLMR we choose a collection of 15 diverse
and widely used complex networks from two popular
archives~\cite{Bader2017a,Leskovec2014a}. These networks are listed
in Table~\ref{tab:extra_social}. For each network and
each comparison, \ie one run of PLM \vs another run of PLM or SLM \vs
SLMR, we get a pair of partitions $\mathcal{P}$ and $\mathcal{P}'$. 

Note that we are not aware of comparable many-to-many correspondences approaches,
so that a comparison to existing methods has to be omitted.
Our sequential code implementing the algorithms presented in
  Section~\ref{sec:BB} is written in C++; it uses the LM implementations
  of NetworKit~\cite{Staudt2014}.
  



\vspace{-2.5ex}
\begin{center}
\begin{table*}[ht]
\caption{Complex networks used for comparing partitions. The column
  {\small \emph{\# communities}} indicates the average number of communities
  generated by PLM (average over 20 runs).}
\label{tab:extra_social}
\begin{center}
\begin{small}
\scalebox{0.78}{
\begin{tabular}{ l | l | r | r | r | r }
Graph ID & Name & \#vertices & \#edges & \# communities & Network type\\ \hline \hline
 1 & \textsc{p2p-Gnutella}          & \numprint{6405}   & \numprint{29215}     & {12.7}    & filesharing network\\\hline
 2 & \textsc{PGPgiantcompo}         & \numprint{10680}  & \numprint{24316}     & {95.7}    & network of PGP users\\\hline
 3 & \textsc{email-EuAll}           & \numprint{16805}  & \numprint{60260}     & {48.4}    & network of connections via email\\\hline
 4 & \textsc{as-22july06}           & \numprint{22963}  & \numprint{48436}     & {26.1}    & autonomous systems in the internet\\\hline
 5 & \textsc{soc-Slashdot0902}      & \numprint{28550}  & \numprint{379445}    & {144.4}    & news network\\\hline
 6 & \textsc{loc-brightkite\_edges} & \numprint{56739}  & \numprint{212945}    & {264.7}    & location-based friendship network\\\hline
 7 & \textsc{loc-gowalla\_edges}    & \numprint{196591} & \numprint{950327}    & {509.7}    & location-based friendship network\\\hline
 8 & \textsc{coAuthorsCiteseer}     & \numprint{227320} & \numprint{814134}    & {181.5}    & citation network\\\hline
 9 & \textsc{wiki-Talk}             & \numprint{232314} & \numprint{1458806}   & {632.3}    & user interactions through edits\\\hline
 10 & \textsc{citationCiteseer}      & \numprint{268495} & \numprint{1156647}  & {124.8}    & citation network\\\hline
 11 & \textsc{coAuthorsDBLP}         & \numprint{299067} & \numprint{977676}   & {181.7}    & citation network\\\hline
 12 & \textsc{web-Google}            & \numprint{356648} & \numprint{2093324}  & {159.0}    & hyperlink network of web pages\\\hline
 13 & \textsc{coPapersCiteseer}      & \numprint{434102} & \numprint{16036720} & {266.9}    & citation network\\\hline
 14 & \textsc{coPapersDBLP}          & \numprint{540486} & \numprint{15245729} & {146.2}    & citation network\\\hline
 15 & \textsc{as-skitter}            & \numprint{554930} &\numprint{5797663}   & {226.8}    & network of internet service providers\\\hline
\end{tabular}}
\end{small}
\end{center}
\end{table*}
\end{center}
\begin{figure*}[bht]
\begin{center}
\subfloat[]{
\includegraphics[width=0.22\textwidth]{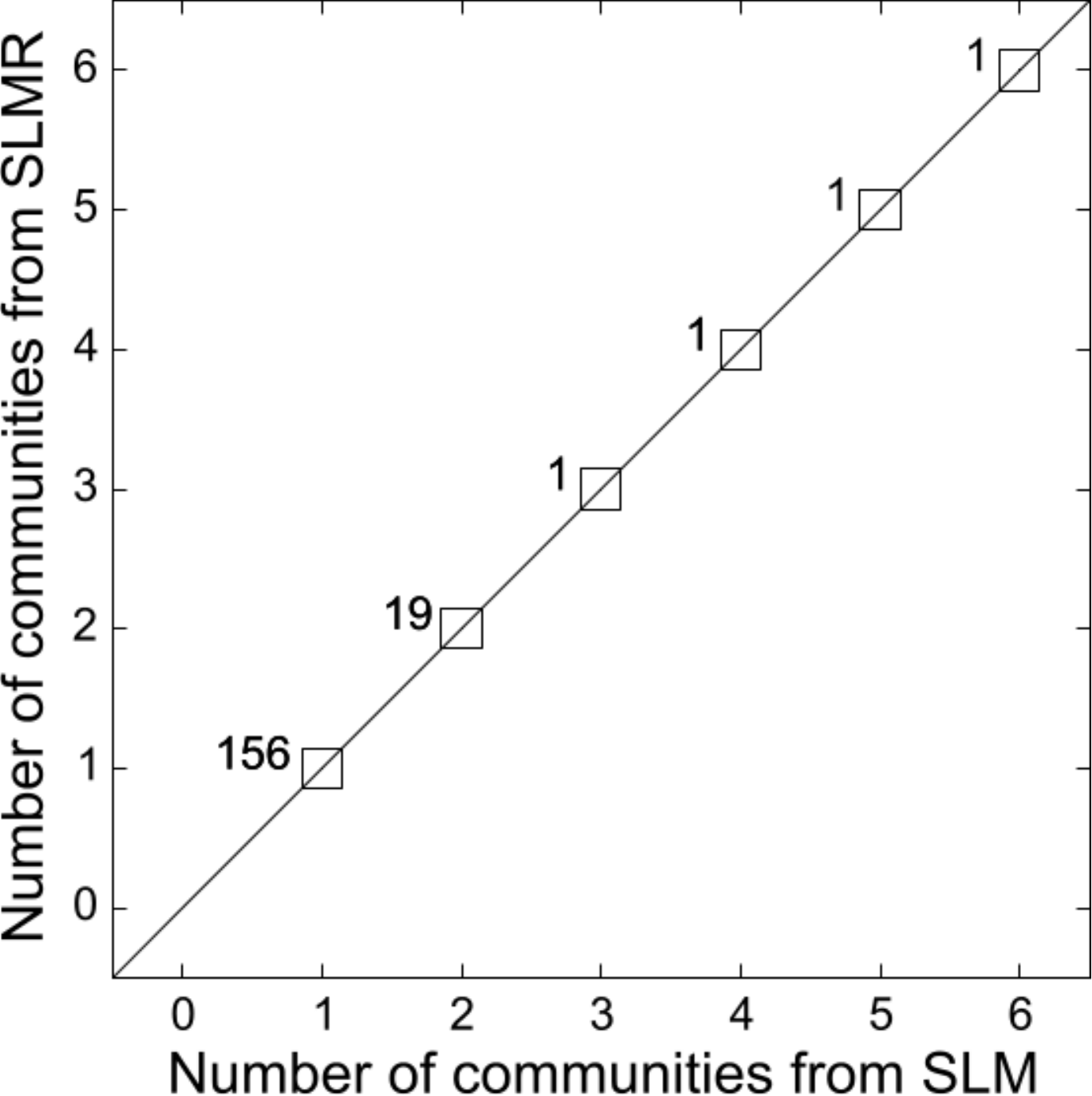}
}\qquad
\subfloat[]{
\includegraphics[width=0.22\textwidth]{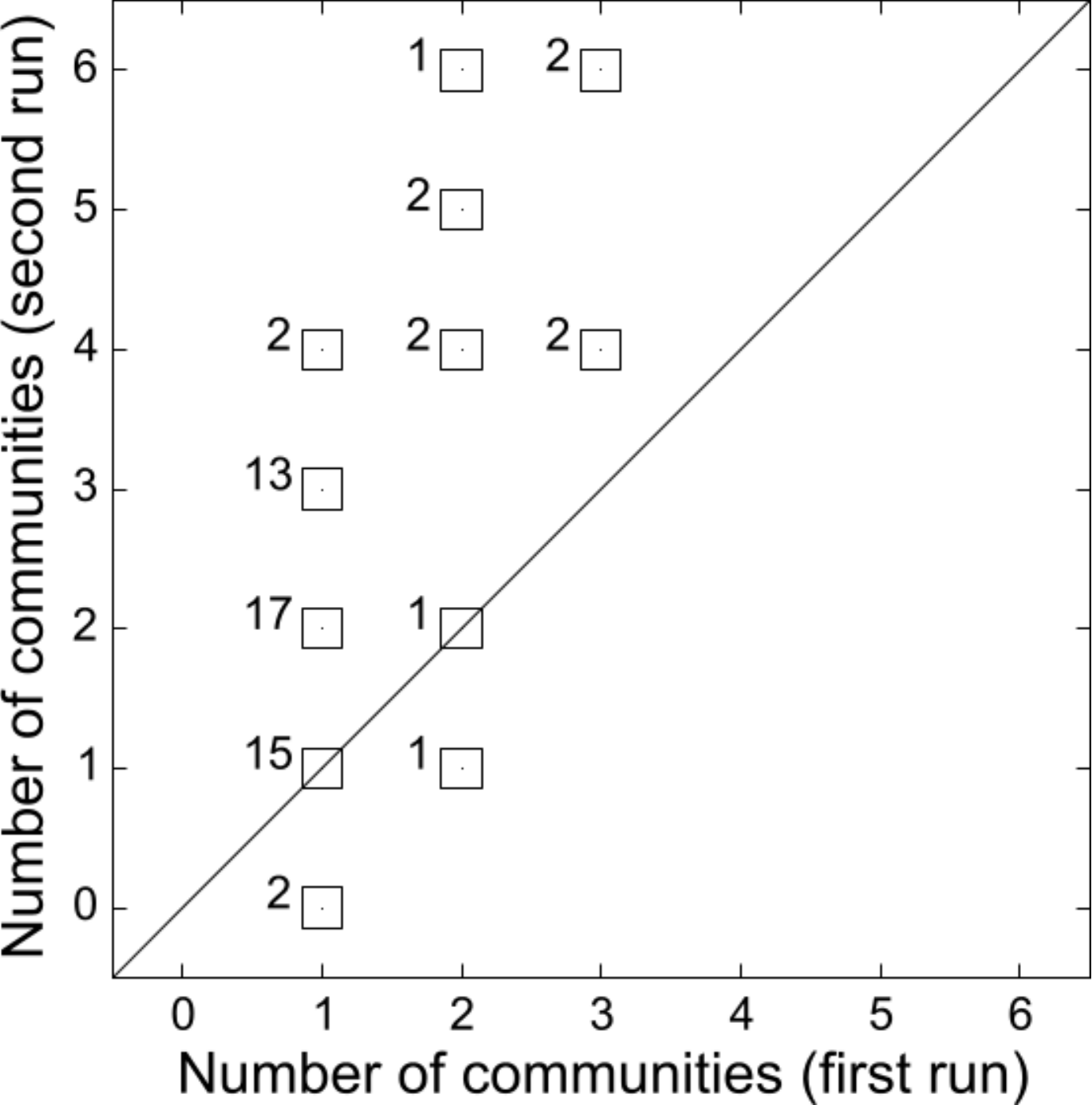}
}\qquad
\subfloat[]{
\includegraphics[width=0.23\textwidth]{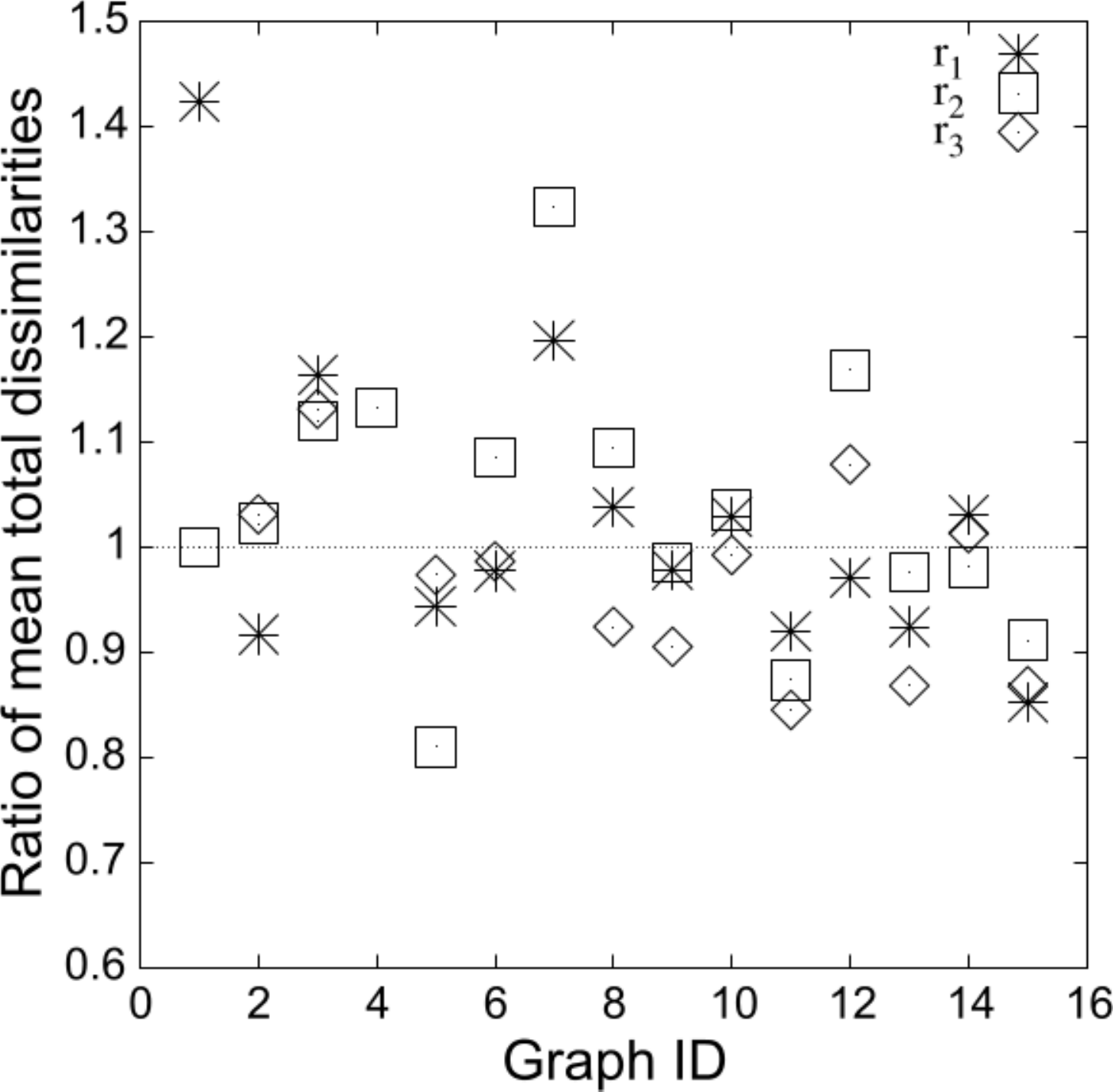}
}
\caption{A data point $(n, n')$ indicates that there are
  \ctwo-correspondences between partitions $\mathcal{P}$ and
  $\mathcal{P}'$, where $n$ communities of $\mathcal{P}$ correspond to
  $n'$ communities of $\mathcal{P}'$. The number at a data point $(n,
  n')$ indicates the number of \ctwo-correspondences $(\mathcal{S},
  \mathcal{S}')$ with $(\vert \mathcal{S} \vert, \vert \mathcal{S}'
  \vert) = (n, n')$. PLM, SLM and SLMR are applied to the graph \web.
  (a) Pairs from SLM \vs SLMR. (b)
  Pairs from different runs of PLM. The pairs not shown are (7,12) and
  (9,19) with one occurrence each. (c) Fluctuation of mean total
  dissimilarities of correspondences over 10 different runs of PLM \vs
  PLM seems to be due to the non-determinism of PLM. 
  }
\label{fig:web}
\end{center}
\vspace{-2ex}
\end{figure*}
%
\subsection{Results.}

For any network in Table~\ref{tab:extra_social}, all $\vert \mathcal{P} \vert - 1$ best
\ctwo-cor\-res\-pon\-den\-ces $(\mathcal{S}, \mathcal{S}')$ between
partitions from runs of SLM and SLMR (both deterministic)
fulfill $\vert \mathcal{S} \vert = \vert \mathcal{S}' \vert$. 
(Figure~\ref{fig:web}(a) shows the corresponding result for \textsc{web-Google}.)
This indicates that, between SLM and SLMR, the communities merely exchange
elements with each other and that there are no unions and no
  break-ups of communities.  Calculations of $r_1$, $r_2$ and $r_3$
for SLM \vs SLMR and all networks yield values between $1.0$ and
$1.024$. Thus, none of the choices, \ie B\&B \vs heuristic
and \ctwo \vs \cthree, has considerable impact on quality.
Figure~\ref{fig:web}(b) shows the results of analogous experiments with
PLM \vs PLM instead of SLM \vs SLMR. In contrast to Figure~\ref{fig:web}(a), numerous
\ctwo-correspondences $(\mathcal{S}, \mathcal{S}')$ are unbalanced in
that $\vert \mathcal{S} \vert$ and $\vert \mathcal{S}' \vert$ differ
considerably. This indicates that the non-determinism of PLM causes
unions and break-ups of communities.

Figure~\ref{fig:web}(c) shows that the fluctuations of total
  dissimilarity (after some averaging) do not follow any trend in
  terms of the four scenarios. Since the B\&B algorithm
  cannot perform worse than the corresponding heuristic on a given
  partition, this indicates that the fluctuations are due to the
  non-determinism of PLM, and that the heuristics are as good as the
  corresponding B\&B algorithm.

Running times of our B\&B algorithm fluctuate considerably, \eg
between 38 and 9555 seconds in ten runs for the graph
\textsc{wiki-Talk}.  Minimum, maximum and mean running times for
graphs in Table~\ref{tab:extra_social} are shown in
Table~\ref{tab:running_BB} (Appendix~\ref{sec:app:run}). Recall that running times refer to
computing the best $\vert \mathcal{P} \vert -1$ correspondences. Not
surprisingly, running times tend to increase enormously with
increasing numbers of communities, despite the strong fluctuations.
In the vast majority of cases, however, the B\&B algorithm terminates
within a few minutes, even for the larger instances.
The analogue running times of our greedy heuristic are much more
stable, never exceeding 40 seconds; for details see
Table~\ref{tab:running_greed} (Appendix~\ref{sec:app:run}). As expected, due to the absence of
backtracking, the trend toward higher running times for increasing
numbers of communities is less pronounced than for the B\&B algorithm.
Nonetheless, as mentioned above, the aggregated quality ($r_i \in
[1.0; 1.024]$) shows that $\greed$ yields very good results already.



To summarize, non-de\-ter\-mi\-nism of PLM disrupts the communities in
a more fundamental way (frequent unions or break-ups of communities)
than the refinement phase. 
Also, the choices (i) B\&B \vs
heuristic and (ii) \ctwo \vs \cthree have a minor impact on the
quality of the correspondences. Most of the time the B\&B algorithm is
fast (less than one minute), but outliers with running times of a few
hours do exist. In the context of community detection, however, it
suffices to run $\greed$, which yields very good correspondences
quickly in all cases. 
Another option would be to terminate the
B\&B algorithm after a certain amount of time, taking the best result
found. 

\section{Related work}
\label{sec:app:soa}

\subsection{Similarity measures for partitions.}
Wagner and Wagner~\cite{Wagner2007a} provide a comprehensive
collection of similarity measures for partitions $\mathcal{P}$ and $\mathcal{P}'$ of the same set $V$. They can
all be derived from the contingency table of $\mathcal{P}$ and
$\mathcal{P}'$.
Ref.~\cite{Wagner2007a} groups the
similarity measures into three groups:

\begin{enumerate}
\item Measures based on considering all unordered pairs $\{v, w\}$ of $V$ and counting the 4 cases arising from the
  distinction as to whether $v$ and $w$ belong to the same part or to
  different parts of $\mathcal{P}$ and the analogous distinction with
  $\mathcal{P}'$ instead of $\mathcal{P}$. Examples of such measures
  are the Rand index~\cite{Rand1971a} and the adjusted Rand
  index~\cite{Hubert1985a}.
\item Measures that involve a sum over maximum $P_i$, $P'_j$ overlaps,
  where the sum is over the $P_i$, the maximum is over the $P'_j$, and
  the overlaps are defined in various ways. One example is the
  $\mathcal{F}$-measure~\cite{Larsen1999a,Fung2008a}. Typically,
  theses measures yield different results if the roles of
  $\mathcal{P}$ and $\mathcal{P}'$ are exchanged. The set function
  $\phi_{\mathcal{P}'}(\cdot)$ defined in this paper, see
  Eqs.~(\ref{eq:Frac}) and~(\ref{eq:peak}), has similar
  properties in that it (i) aggregates $P_i$, $P'_j$ overlaps over
  certain $P_i$ in a nonlinear way, and (ii) may vary if the roles of
  $\mathcal{P}$ and $\mathcal{P}'$ are exchanged.
\item Measures that involve mutual information, \eg Normalized Mutual
  Information~\cite{Strehl2003a}. Here, the common ground with our
  approach to defining correspondences is that we can replace the
  function $\peak(\cdot)$, see Eqs.~(\ref{eq:Frac})
  and~(\ref{eq:peak}), by the binary entropy function without altering
  the nature of our optimization problem.
\end{enumerate}

\subsection{Impurity measures.}
The value $\phi_{\mathcal{P}'}(\mathcal{S})$ indicates
how well the parts of $\mathcal{P}'$ fit into $U_{\mathcal{S}}$ or $V
\setminus U_{\mathcal{S}}$, see Eqs.~(\ref{eq:Frac})
and~(\ref{eq:peak}). Impurity measures, as defined
in~\cite{Tan2005a,Simovici2002a}, seem to be based on a similar
idea. Using our setting and notation, Simovici
\etal~\cite{Simovici2002a} define the impurity of a subset $L$ of the
ground set $V$ relative to $\mathcal{P}$ and generated by
$\peak(\cdot)$ as
%
$\IMP^{\peak}_{\mathcal{P}'}(L) = \vert L \vert \sum_{P' \in \mathcal{P}'} \peak(\frac{\vert L \cap P' \vert}{\vert L \vert})$.
 
We can turn $\IMP^{\peak}_{\mathcal{P}'}(U_{\mathcal{S}})$ into
$\phi_{\mathcal{P}'}(\mathcal{S})$ by (i) pulling
$U_{\mathcal{S}}$ under the sum (mathematically correct) and (ii)
exchanging the roles of $U_{\mathcal{S}}$ and $P'$ under the sum
(mathematically incorrect). For us it is important to have the roles
as they are in $\phi_{\mathcal{P}'}(\mathcal{S})$ because
this is what makes $\phi_{\mathcal{P}'}(\cdot)$ a
submodular and symmetric function. These properties, in turn, make it
possible to find the best nontrivial $\mathcal{S}$ in polytime.

Despite this mismatch between $\IMP^{\peak}_{\mathcal{P}'}(\cdot)$ and
$\phi_{\mathcal{P}'}(\cdot)$, studying
$\IMP^{\peak}_{\mathcal{P}'}(\cdot)$ helped to develop the intuition
behind our approach.
The main properties of $\IMP^{\peak}_{\mathcal{P}'}(\cdot)$ and
related measures, as formulated and proven in~\cite{Simovici2002a},
are preserved if $\peak(\cdot)$ is replaced by another
\emph{generator}, as defined in~\cite{Simovici2002a}, \ie another
concave and subadditive function $f : [0,1] \mapsto \RR$ with $f(0) =
f(1) = 0$. Likewise, if we replace $\peak(\cdot)$ in
Eq.~(\ref{eq:Frac}) by another generator, we will arrive at
similar definitions of correspondences. This also does not change the
kind and asymptotic complexity of the optimization problems posed by
our approach.

%
%
\section{Conclusions and outlook}
Recall that small data changes can lead clustering methods to split
or merge clusters. By computing many-to-many correspondences,
one can recover the most crucial split and merge operations.
Here, \ctwo-correspondences are ideal
in that the many-to-many correspondences and the associated split and
merge operations make up a hierarchy.
For \ctwo-correspondences there exists a minimum basis of
\emph{non-cros\-sing} $P_s$-$P_t$ cuts of $\mathcal{P}$ \wrt $\mathcal{P}'$ that, in turn, yield a hierarchy of the $\vert
\mathcal{P} \vert - 1$ best correspondences between $\mathcal{P}$ and
$\mathcal{P}'$ via optimal partners.

Under \cthree,
the cuts in a minimum cut basis
are crossing in general. 
On the upside and in contrast to
\ctwo-correspondences, 
a good \cthree-correspondence gives rise to two
similar cuts $(U_{\mathcal{S}}, U_{\mathcal{P} \setminus
  \mathcal{S}})$ and $(U_{\mathcal{S}'}, U_{\mathcal{P}' \setminus
  \mathcal{S}'})$ of $V$. For such a pair of similar cuts it is
easy to find a cut that mediates between them. The overlay
of $k$ such medial cuts then results in a consensus partition.


In our B\&B algorithm, one has to choose the next candidates for
extension of either $\mathcal{S}_s$ or $\mathcal{S}_t$. This choice
may also involve application-specific criteria such as color or shape
in image analysis. Such additional information may help our B\&B
algorithm to stay in the lane, and is expected to accelerate it.

We see our B\&B algorithm as a starting point for 
fast heuristics to find high-quality correspondences. In the
  experiments of Section~\ref{sec:real-sim} we have seen that turning
  off the backtracking in our B\&B algorithm (\ie running $\greed$ only once)
  has only a negligible effect on the quality of the results. 

\begin{small}
\vspace{1.5ex}
\textbf{Acknowledgements.}
We thank Christian Staudt for helpful discussions and the anonymous
reviewers for helping to improve the paper in various respects.
\end{small}


\bibliographystyle{abbrv}
\bibliography{roland,refs-parco}

%
%

\clearpage

\appendix
\section{Appendix}
\label{sec:app}
%

\subsection{Correspondences in image analysis.}
\label{subsec:app:image analysis}

\pagenumbering{roman}
\setcounter{page}{1}

Scenarios in which it makes sense to compare
$\mathcal{P}$ and $\mathcal{P}'$ using correspondences can be as follows:
(i) $\mathcal{P}$ is the result of a segmentation algorithm and
  $\mathcal{P}'$ describes ground truth, \eg if $\mathcal{P}$ is a
  segmented satellite image and if $\mathcal{P}'$ describes land use
  that has been determined in the field by experts. Here, the aim of a
  comparison might be to identify areas where $\mathcal{P}$ suffers
  from over-segmentation (unions of regions of $\mathcal{P}$ that
  correspond well to single regions of $\mathcal{P}'$), from
  under-segmentation (single regions of $\mathcal{P}$ that
  correspond well to unions of regions of $\mathcal{P}'$) or
  more intricate combinations of over-segmentation and
  under-segmentations.
 
(ii) $\mathcal{P}$ and $\mathcal{P}'$ describe ground truth at
  different times. Sticking to land use, a good
  correspondence $(\mathcal{S}, \mathcal{S}')$ with $\vert \mathcal{S}
  \vert, \vert \mathcal{S}' \vert > 1$ may indicate crop rotation.

(iii) $\mathcal{P}$ and $\mathcal{P}'$ are results of different
  segmentation algorithms applied to the same image, and/or the two
  segmentations are based on different physical measurements, \eg
  channels in Satellite Imagery or CT \vs MRI in medical
  imaging. Then, a good correspondence $(\mathcal{S}, \mathcal{S}')$
  provides strong evi\-dence that the feature described by $\mathcal{S}$
  is not an artifact. For an example of correspondences between
  different segmentations see Figure~\ref{fig:sheep}.


\FloatBarrier

\subsection{Examples of cuts and correspondences.}
\label{sub:exmp-cuts-corres}
Illustration for example from Section~\ref{sec:problem-statement}:
\begin{figure}[h!]
\begin{center}
\includegraphics[width=0.9\columnwidth]{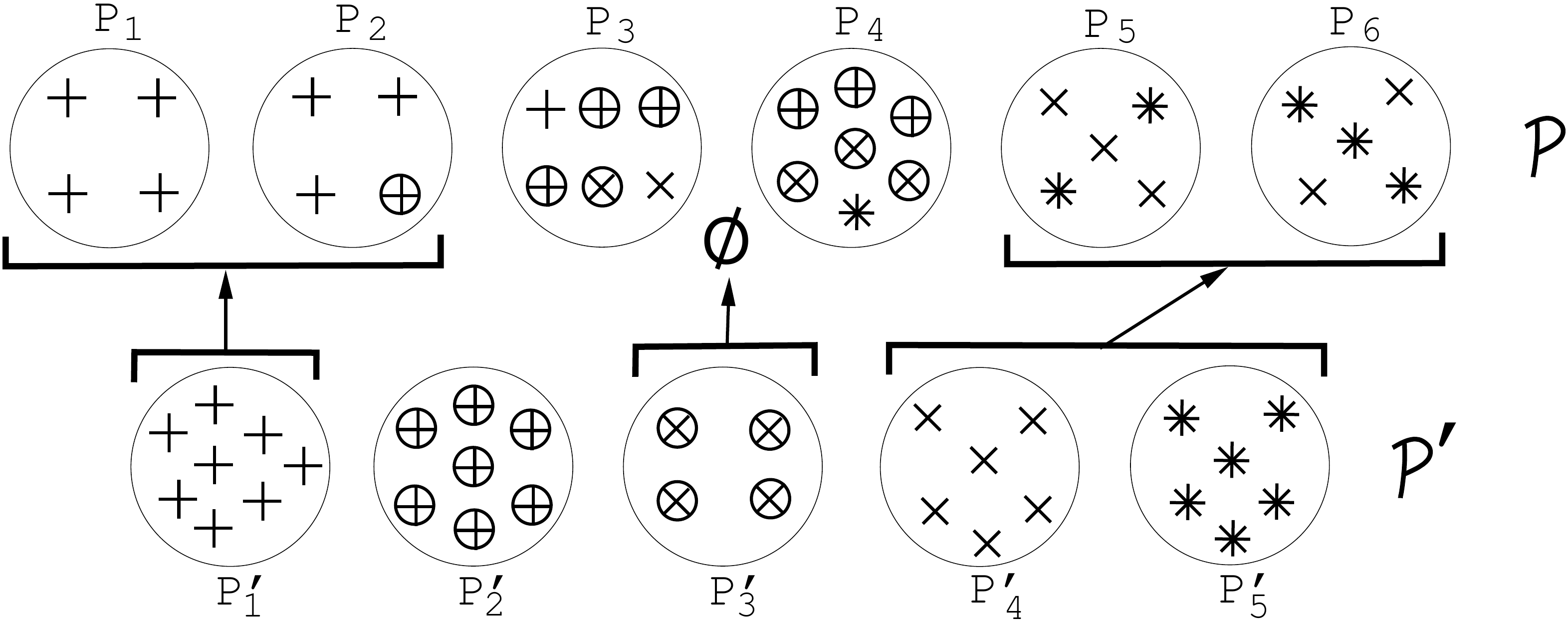}
\caption{Same scenario as in Figure~\ref{fig:example} with the roles of
  $\mathcal{P}$ and $\mathcal{P}'$ exchanged.}
\label{fig:revExample}
\end{center}
\end{figure}
\vspace{-4ex}

\subsection{\cone-correspondences.}
\label{sec:app:goodC1}
We show how a natural collection of $\vert \mathcal{P} \vert + \vert
\mathcal{P}' \vert - 1$ best \cone-correspondences emerges from a
minimum cut basis of an edge-weighted bipartite graph $G=(\mathcal{P}
\sqcup \mathcal{P}', E, \omega(\cdot))$, where the edge set $E$
consists of all $\{P, P'\}$ with $P \in \mathcal{P}$, $P' \in
\mathcal{P}'$ and $P \cap P' \neq \emptyset$. The edge weights are
given by $\omega(\{P, P'\}) = \vert P \cap P' \vert$. The cuts in the
basis may be chosen such that they are non-crossing, see
Section~\ref{subsec:bipartite}.  Time complexities for
finding good \cone-cor\-res\-pon\-den\-ces are discussed in
Section~\ref{subsec:goodC1_running_times}.

\subsubsection{\cone-correspondences from cuts of
a bipartite graph}
\label{subsec:bipartite}
Recall that \cone: $\mathcal{S} \notin \{\emptyset, \mathcal{P}\} \vee
\mathcal{S}' \notin \{\emptyset, \mathcal{P}'\}$ is our weakest
constraint. It merely excludes trivial and very bad correspondences.
Finding good correspondences through finding small cuts of certain
bipartite graphs has already been proposed in the context of mutual
document and word clustering~\cite{Zha2001a}. Here, we start by
rewriting $\phi(\mathcal{S}, \mathcal{S}')$.
\begin{align*}
\phi(\mathcal{S}, \mathcal{S}') &= \vert
U_{\mathcal{S}} \setminus U_{\mathcal{S}'} \vert + \vert
U_{\mathcal{S}'} \setminus U_{\mathcal{S}} \vert \nonumber\\ &= \vert
U_{\mathcal{S}} \cap (V \setminus U_{\mathcal{S}'})\vert + \vert
U_{\mathcal{S}'} \cap (V \setminus U_{\mathcal{S}})\vert
\nonumber\\ &= \sum_{P' \notin \mathcal{S}'} \vert U_{\mathcal{S}}
\cap P'\vert + \sum_{P \notin \mathcal{S}} \vert U_{\mathcal{S}'} \cap
P \vert \nonumber\\ &= \sum_{P' \notin
  \mathcal{S}'}(\sum_{P \in \mathcal{S}} \vert P \cap P' \vert) + \sum_{P \notin
  \mathcal{S}}(\sum_{P' \in \mathcal{S}'} \vert P \cap P' \vert).
\label{eq:bipart}
\end{align*}

Let $G=(W, E, \omega(\cdot))$ with $\omega:E \mapsto \RR_{\geq 0}$ be
the edge-weighted bipartite graph defined by (i) $W := \mathcal{P}
\sqcup \mathcal{P}'$, where $\sqcup$ denotes the disjoint union, (ii)
$E := \{\{P, P'\} \mbox{~with~} P \in \mathcal{P}, P' \in \mathcal{P}'
\mbox{~and~} P \cap P' \neq \emptyset\}$ and (iii) $E := \{\{P, P'\}
\mbox{~with~} P \in \mathcal{P}, P' \in \mathcal{P}' \mbox{~and~} P
\cap P' \neq \emptyset\}$ and $\omega(\{P, P'\}) := \vert P \cap P'
\vert$,
where we distinguish between $P \in \mathcal{P}$ and $P' \in
\mathcal{P}'$, even if $P = P'$. Then, $\phi(\mathcal{S},
\mathcal{S}')$ equals the total weight of the cut $(\mathcal{S} \sqcup
\mathcal{S}', (\mathcal{P} \setminus \mathcal{S}) \sqcup (\mathcal{P}'
\setminus \mathcal{S}'))$.

Thus, a minimum cut basis of $G$ gives rise to a minimum basis of
\cone-correspondences. A minimum cut basis of $G$, in turn, can be
chosen such that the cuts are non-crossing~\cite{Gomory1961a}. Hence,
a minimum basis of \cone-correspondences can be represented by a
Gomory-Hu tree with vertex set $W$. 

\subsubsection{Asymptotic time for minimum cut basis of \cone-correspondences}
\label{subsec:goodC1_running_times}
To build a minimum cut basis of \cone-correspondences, we first
generate the bipartite graph $G$. To this end, we compute the
contingency table of $\mathcal{P}$ and $\mathcal{P}'$ (weighted
adjacency matrix of $G$), \ie the matrix whose entry at $(i,j)$ equals
$\vert P_i \cap P'_j \vert$. Initializing the contingency table to
zero entries takes time $\bigO(\vert \mathcal{P} \vert \vert
\mathcal{P}' \vert)$. The contingency table can then be filled in one
traversal of $V$, provided that deciding on the membership of any $v
\in V$ to a part in $\mathcal{P}$ and $\mathcal{P}'$ takes constant
time. Asymptotic time for computing the contingency table \emph{and}
building $G$ is the same as for just computing the contingency table,
\ie $\bigO(\vert V \vert + \vert \mathcal{P} \vert \vert \mathcal{P}'
\vert)$.

Then, based on $G$ and using the algorithms by Gomory or
Gusfield~\cite{Gomory1961a,Gusfield90a}, one can compute a minimum cut
basis of $G$. The asymptotic time of both algorithms amounts to that
of $\vert W \vert - 1$ calculations of minimum $Q$-$R$ cuts, $Q, R \in
W$. Given $Q, R \in W$, a minimum $Q$-$R$ cut can be found in
$\bigO(\vert W \vert \vert E \vert)$ time using an algorithm
in~\cite{Orlin2013a}, which also works for general $G$. Alternatively,
one can use an algorithm in~\cite{Ahuja1994a} which finds a minimum
$Q$-$R$ cut of $G$ in time $\bigO(\mu \vert E \vert \log(2 + \mu^2 /
\vert E \vert))$, where $\mu = \min\{\vert \mathcal{P} \vert, \vert
\mathcal{P}' \vert\}$. The latter algorithm makes use of $G$ being
bipartite and can yield a lower asymptotic time than the former if (i)
$(\vert \mathcal{P} \vert \ll \vert \mathcal{P}' \vert) \vee (\vert
\mathcal{P}' \vert \ll \vert \mathcal{P} \vert)$ and (ii) $G$ is
sparse. Proposition~\ref{prop:running-nontrivial:G} summarizes the
running times of the two algorithms and expresses them in our terms,
\ie $V$, $\mathcal{P}$ and $\mathcal{P}'$.

\begin{prop}
\label{prop:running-nontrivial:G}
A minimum cut basis of $G$ can be computed in $\bigO(\vert V \vert +
\vert \mathcal{P} \vert^3 \vert \mathcal{P}' \vert + \vert \mathcal{P}
\vert \vert \mathcal{P}' \vert^3)$ or in
\begin{equation*}
\bigO(\vert V \vert + \vert \mathcal{P} \vert^2 \vert \mathcal{P}'
\vert^2 \log (2 + \frac{(\min\{\vert \mathcal{P} \vert, \vert
  \mathcal{P}' \vert\})^2}{\max\{\vert \mathcal{P} \vert, \vert
  \mathcal{P}' \vert\}})).
\end{equation*}
\end{prop}
\begin{proof}
Recall that generating $G$ takes time $\bigO(\vert V \vert + \vert
\mathcal{P} \vert \vert \mathcal{P}' \vert)$. Total time is the sum of
the latter and time for $\vert W \vert - 1$ calculations of minimum
$Q$-$R$ cuts.

Using the algorithm in~\cite{Orlin2013a}, $\vert W \vert - 1$
calculations of minimum $Q$-$R$ cuts take time $\bigO(\vert W \vert^2
\vert E \vert) = \vert \mathcal{P} \vert^3 \vert \mathcal{P}' \vert +
\vert \mathcal{P} \vert^2 \vert \mathcal{P}' \vert^2 + \vert
\mathcal{P} \vert \vert \mathcal{P}' \vert^3 = \vert \mathcal{P}
\vert^3 \vert \mathcal{P}' \vert + \vert \mathcal{P} \vert \vert
\mathcal{P}' \vert^3$. This yields the first asymptotic time.

To see that the second asymptotic time is valid, note that (i) $\vert
E \vert \leq \vert \mathcal{P} \vert \vert \mathcal{P}' \vert$ and
(ii) $\vert E \vert \geq \max\{\vert \mathcal{P} \vert, \vert
\mathcal{P}' \vert\}$. The remainder of the proof is
straightforward.
\end{proof}

\subsection{Proofs.}

\subsubsection{Proof of Proposition~\ref{prop:eq:firstMin1}.}
\label{proof:eq:firstMin1}
\begin{proof}
\noindent Starting with Eq.~(\ref{eq:firstMin1}), we get 

\begin{align}
\phi(\mathcal{S}, \mathcal{S}') &= \vert U_{\mathcal{S}} \setminus
U_{\mathcal{S}'} \vert + \vert U_{\mathcal{S}'} \setminus
U_{\mathcal{S}} \vert\nonumber\\ &= \vert U_{\mathcal{S}} \cap (V
\setminus U_{\mathcal{S}'})\vert + \sum_{P' \in \mathcal{S}'} \vert P'
\setminus U_{\mathcal{S}}\vert\nonumber\\ &= \sum_{P' \notin
  \mathcal{S}'} \vert U_{\mathcal{S}} \cap P'\vert + \sum_{P' \in
  \mathcal{S}'}(\vert P' \vert - \vert U_{\mathcal{S}} \cap P'
\vert)\label{eq:secondMin}.
\end{align}

\noindent By letting $\mathcal{S}'$ be an optimal partner of
$\mathcal{S}$, \eg by calculating $\mathcal{S}'$ using
Eq.~(\ref{eq:partner}), we minimize the contribution (damage) of
each $P' \in \mathcal{P}'$ to the right hand side of
Eq.~(\ref{eq:secondMin}), and thus minimize $\phi(\mathcal{S},
\cdot)$. Insertion of $\mathcal{S}'$ from Eq.~(\ref{eq:partner})
then yields

\begin{small}
\begin{eqnarray}
\min_{\mathcal{S}' \subseteq \mathcal{P}'}
\phi(\mathcal{S},\mathcal{S}') =& \sum_{P' \in \mathcal{P}'}
\min\{\vert U_{\mathcal{S}} \cap P'\vert, \vert P' \vert - \vert
U_{\mathcal{S}} \cap P' \vert\}\nonumber\\ =& \sum_{P' \in
  \mathcal{P}'} \vert P' \vert \min\{\frac{\vert U_{\mathcal{S}} \cap
  P'\vert}{\vert P' \vert}, 1 - \frac{\vert U_{\mathcal{S}} \cap P'
  \vert}{\vert P' \vert}\}\nonumber\\ =& \sum_{P' \in \mathcal{P}'}
\vert P' \vert \peak(\frac{\vert U_{\mathcal{S}} \cap P' \vert}{\vert
  P' \vert}).
\label{eq:towardPeak}
\end{eqnarray}
\end{small}
\end{proof}

\subsubsection{Proof of Proposition~\ref{prop:submodular}.}
\label{proof:prop:submodular}
\begin{proof}
The symmetry of $\phi_{\mathcal{P}'}(\cdot)$ follows from
that of $\peak(\cdot)$. Sums and multiples of submodular functions are
submodular~\cite{Schrijver2003a}. Thus, to show that
$\phi_{\mathcal{P}'}(\cdot)$ is submodular, it suffices to
show that $\phi_i(\mathcal{S}) := \peak(\frac{\vert U_{\mathcal{S}}
  \cap P'_i \vert}{\vert P'_i \vert})$ in Eq.~(\ref{eq:Frac}) is
submodular for all $i$.

Indeed, the $\phi_i(\cdot)$ are of the form $c(m(\cdot))$, where
$c(\cdot)$ is concave and $m(\cdot)$ is non-negative modular. Any
function of this form is
submodular~\cite{Stobbe2010b,Bilmes2012a}.
\end{proof}

\subsubsection{Proof of Proposition~\ref{prop:running-nontrivial:P}.}
\label{proof:prop:run-time-ctwo}
\begin{proof}
By means of the bipartite graph $G$ in
  Section~\ref{subsec:bipartite}, finding a mi\-ni\-mum $P_s$-$P_t$
  cut $(\mathcal{S}, \mathcal{P} \setminus \mathcal{S})$ of
  $\mathcal{P}$ can be achieved through (i) finding a minimum
  $P_s$-$P_t$ cut $(\mathcal{S} \sqcup \mathcal{S}', (\mathcal{P}
  \setminus \mathcal{S}) \sqcup (\mathcal{P}' \setminus
  \mathcal{S}'))$ of $G$ and (ii) extracting $(\mathcal{S},
  \mathcal{P} \setminus \mathcal{S})$. Step (i) is analogous to the
  proof of Proposition~\ref{prop:running-nontrivial:G}.
\end{proof}

\subsubsection{Proof of Proposition~\ref{prop:complexity}.}
\label{proof:prop:complexity}
\begin{proof}
Below, we refer to Algorithm OPTIMAL-SET
from~\cite{Queyranne98a}, where a symmetric submodular function
$f(\cdot)$ is minimized by building the set minimizing $f(\cdot)$ from
scratch. OPTIMAL-SET consists of $\bigO(\vert \mathcal{P} \vert^3)$
evaluations of $\phi_{\mathcal{P}'}(\cdot)$~\cite[Theorem
  3]{Queyranne98a}.

Due to (i) the fact that all distributions can be computed in
$\bigO(\vert V \vert + \vert \mathcal{P} \vert \vert \mathcal{P}'
\vert)$, (ii) Proposition~\ref{prop:submodular} of this paper, (iii)
Theorem~3 in~\cite{Queyranne98a} (which uses OPTIMAL-SET) and (iv)
Proposition~\ref{prop:evaluateFrac} of this paper, the asymptotic
running time for mi\-ni\-mi\-zing Eq.~(\ref{eq:Frac}) under the
constraint $\mathcal{S} \notin \{\emptyset, \mathcal{P}\}$ amounts to
$\bigO(\vert V \vert + \vert \mathcal{P} \vert^3 \vert \mathcal{P}
\vert \vert \mathcal{P}'\vert) = \bigO(\vert V \vert + \vert
\mathcal{P} \vert^4 \vert \mathcal{P}' \vert)$.
\end{proof}

\subsubsection{Proof of Proposition~\ref{prop:bound}.}
\label{proof:bound}
\begin{proof}
\begin{align}
\phi_{\mathcal{P}'}(\mathcal{S}) &=\sum_{P' \in
  \mathcal{P}'} \vert P' \vert \peak(\frac{\vert U_{\mathcal{S}} \cap
  P' \vert}{\vert P' \vert})\nonumber\\ &= \sum_{P' \in \mathcal{P}'}
\vert P' \vert \min\{\frac{\vert U_{\mathcal{S}} \cap P' \vert}{\vert
  P' \vert}, \frac{\vert U_{\mathcal{P} \setminus \mathcal{S}} \cap P'
  \vert}{\vert P' \vert}\}\nonumber\\ &= \sum_{P' \in \mathcal{P}'}
\min\{\vert U_{\mathcal{S}} \cap P' \vert, \vert U_{\mathcal{P}
  \setminus \mathcal{S}} \cap P' \vert\}\\ &\geq b(\mathcal{S}_s,
\mathcal{S}_t).
\end{align}
\end{proof}

\subsection{\cthree-correspondences and \cfour-correspondences.}
\label{sec:app:C3C4}
We first show that finding an optimal \cthree-correspondence between
$\mathcal{P}$ and $\mathcal{P}'$ amounts to finding a nontrivial
minimum $\mathcal{S}$ of a symmetric and non-submodular function
$\phistar: 2^{\mathcal{P}} \mapsto \RR_{\geq 0}$. Alternatively,
$\mathcal{S}$ can be found through minimizing $\vert \mathcal{P}
\vert^2$ non-symmetric submodular functions. As for \ctwo, a good
$\cthree$-correspondence is essentially a small cut of
$\mathcal{P}$. Since $\phistar$ is symmetric, there exists a minimum
cut basis containing $\vert \mathcal{P} \vert - 1$ minimum $P_s$-$P_t$
cuts of $\mathcal{P}$. These cuts give rise to a natural collection of
$\vert \mathcal{P} \vert - 1$ best \cthree-correspondences.

The rest of the section is on \cfour-correspondences. We derive a
property of \cfour-correspondences which suggests that finding good
\cfour-correspondences is more difficult than submodular
minimization. Nevertheless, the techniques that we developed for
finding good correspondences under the constraints $\ctwo$ and
$\cthree$ may be useful for finding at least a subset of good
\cfour-correspondence in a real-world application.

\paragraph*{\cthree-correspondences.}
The constraint \cthree: $\mathcal{S} \notin \{\emptyset, \mathcal{P}\}
\wedge \mathcal{S}' \notin \{\emptyset, \mathcal{P}'\}$ ensures that
$(U_{\mathcal{S}}, U_{\mathcal{P} \setminus \mathcal{S}})$ and
$(U_{\mathcal{S}'}, U_{\mathcal{P}' \setminus \mathcal{S}'})$ are cuts
of $V$. This is a prerequisite for finding a consensus partition via
good correspondences, see Section~\ref{subsec:constraints}.
Finding an optimal \cthree-correspondence $(\mathcal{S},
\mathcal{S}')$ amounts to finding $\emptyset \neq \mathcal{S}
\subsetneq \mathcal{P}$ with a minimum value of
$\phistar~:~2^{\mathcal{P}} \mapsto
\RR_{\geq 0}$ defined as
\begin{equation*}
\label{eq:phistar}
\phistar(\mathcal{S}) := \left\{
\begin{array}{ll}
0 & \mbox{if } \mathcal{S} \in \{\emptyset, \mathcal{P}\},\\
\min_{\emptyset \neq \mathcal{S}' \subsetneq \mathcal{P}'} \vert
U_{\mathcal{S}} \triangle U_{\mathcal{S}'}\vert & \mbox{otherwise.}
\end{array}
\right.
\end{equation*}

\begin{prop} $\phistar(\cdot)$
is symmetric and not submodular.
\label{prop:submodular-nondeg}
\end{prop}
\begin{proof}
If $\mathcal{S} \in \{\emptyset, \mathcal{P}\}$, then
$\phistar(\mathcal{P} \setminus \mathcal{S}) = \phistar(\mathcal{S}) =
0$. Otherwise,
\begin{align*}
\phistar(\mathcal{P} \setminus \mathcal{S})
&= \min_{\emptyset \neq \mathcal{S}' \subsetneq \mathcal{P}'} \vert
U_{\mathcal{P} \setminus \mathcal{S}} \triangle
U_{\mathcal{S}'}\vert\\ &= \min_{\emptyset \neq \mathcal{S}'
  \subsetneq \mathcal{P}'} \vert U_{\mathcal{P} \setminus \mathcal{S}}
\triangle U_{\mathcal{P}' \setminus \mathcal{S}'}\vert\\ &=
\min_{\emptyset \neq \mathcal{S}' \subsetneq \mathcal{P}'} \vert
U_{\mathcal{S}} \triangle U_{\mathcal{S}'}\vert =
\phistar(\mathcal{S}).
\end{align*}

\noindent A counterexample to submodularity of
$\phistar(\cdot)$ is provided in
Figure~\ref{fig:exampleG}.
\end{proof}

\begin{figure}
\begin{center}
\includegraphics[width=0.45\columnwidth]{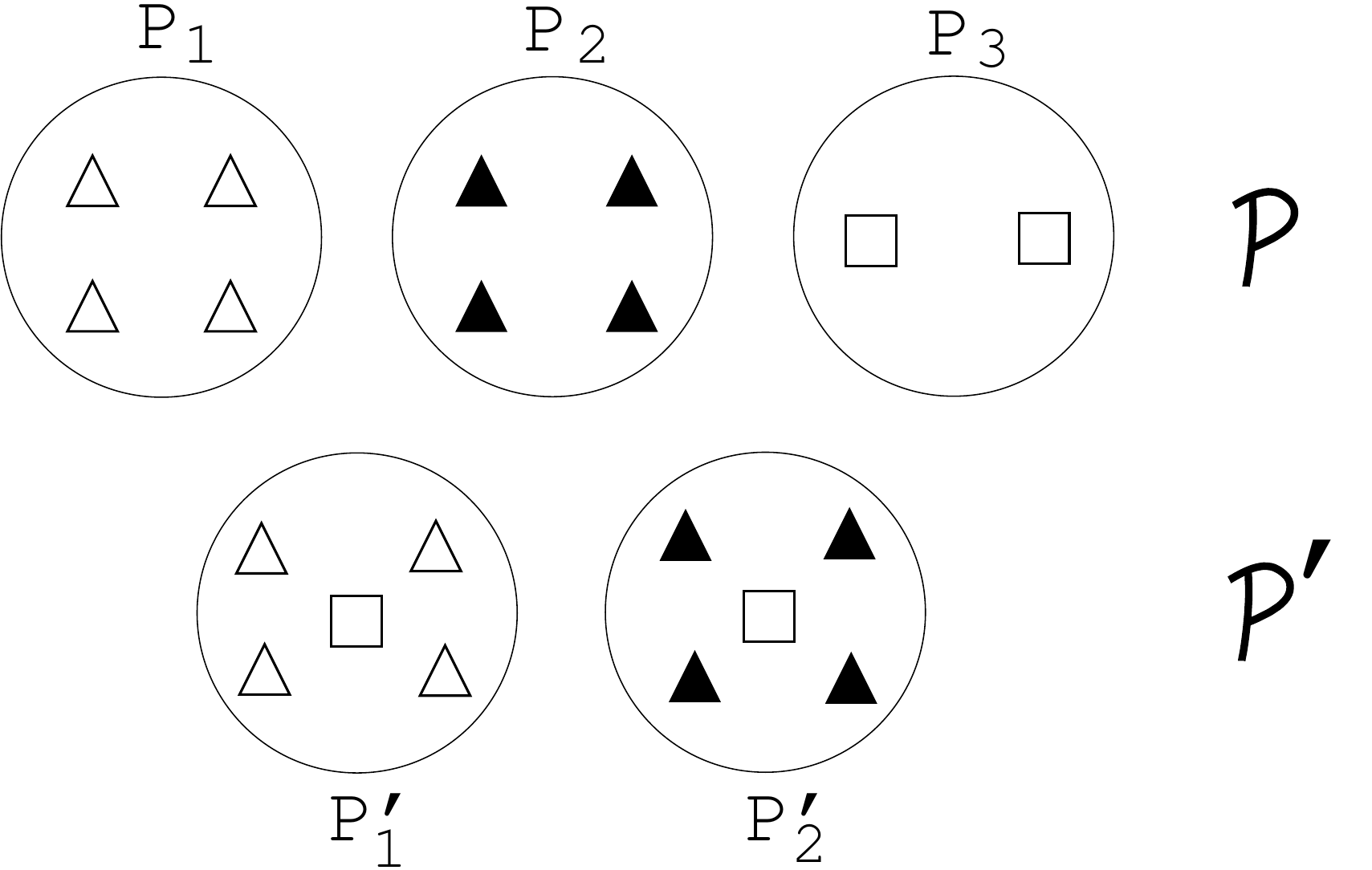}
\caption{\textsc{Counterexample to submodularity of}
  $\phistar(\cdot)$: set $\mathcal{S} :=
  \{P_1, P_3\}$ and $\mathcal{T} := \{P_2, P_3\}$. Then, $(\mathcal{S}
  \cup \mathcal{T})' = \mathcal{P}'$, $\mathcal{S}' = \{P_1'\}$,
  $\mathcal{T}' = \{P_2'\}$ and $(\mathcal{S} \cap \mathcal{T})' =
  \{P_1'\}$ are optimal nontrivial partners of $\mathcal{S} \cup
  \mathcal{T}$, $\mathcal{S}$, $\mathcal{T}$ and $\mathcal{S} \cap
  \mathcal{T}$, respectively. Thus,
  $\phistar(\mathcal{S} \cup \mathcal{T}) =
  0 > \phistar(\mathcal{S}) +
  \phistar(\mathcal{T}) -
  \phistar(\mathcal{S} \cap \mathcal{T}) =
  1 + 1 - 5$.\vspace{0.2cm}\protect\linebreak \textsc{Counterexample
    to symmetry of} $\phistar_{i',j'}(\cdot)$: for $\mathcal{S} =
  \{P_1\}$ we have $\phistar_{1,2}(\mathcal{S}) = \vert P_1 \cap P_2'
  \vert + \vert P_1' \vert - \vert P_1 \cap P_1' \vert = 0 + 5 - 4
  \neq 4 + 5 - 0 = \vert P_1 \cap P_1' \vert + \vert P_2' \vert -
  \vert P_1 \cap P_2' \vert =
  \phistar_{2,1}(\mathcal{S})$.\vspace{0.2cm}\protect\linebreak
  \textsc{Counterexample to mutual correspondences inducing a lattice
    family on $\mathcal{P}$}: The correspondences $(\{P_1\},
  \{P_1'\})$ and $(\{P_2\}, \{P_2'\})$ are mutual, but there is no
  mutual cor\-res\-pon\-dence $(\{P_1, P_2\}, X')$ with $X' \subseteq
  \mathcal{P}'$.}
\label{fig:exampleG}
\end{center}
\end{figure}

Analogous to \ctwo-correspondences, we reformulate the problem of
finding an optimal \cthree-correspondence under the constraint $P_i'
\in \mathcal{S}'$ and $P_j' \in \mathcal{P}' \setminus
\mathcal{S}'$. Let $\mathcal{S} \subseteq \mathcal{P}$ and
$\mathcal{S}' \subseteq \mathcal{P}'$. Then,
Eq.~(\ref{eq:secondMin}) and the constraint imply
\begin{align}
\vert U_{\mathcal{S}} \triangle U_{\mathcal{S}'} \vert = \vert
U_{\mathcal{S}} \cap P_j' \vert + \vert P_i' \vert - \vert
U_{\mathcal{S}} \cap P_i' \vert + \nonumber\\
\sum_{\substack{P'
    \notin \mathcal{S}'\\P' \notin \{P_i', P_j'\}}} \vert
U_{\mathcal{S}} \cap P'\vert + \sum_{\substack{P' \in \mathcal{S}'\\P'
    \notin \{P_i', P_j'\}}}(\vert P' \vert - \vert U_{\mathcal{S}}
\cap P' \vert)
\label{eq:thirdMin}
\end{align}

Analogous to Eq.~(\ref{eq:towardPeak}) we set
\begin{equation}
\label{eq:partner1Copy}
\mathcal{S}' := \{P_i'\} \cup \{P' \in \mathcal{P}' \setminus
\{P_i', P_j'\} : \vert U_{\mathcal{S}} \cap P' \vert > \frac{\vert
  P'\vert}{2}\},
\end{equation}
and thus minimize the contribution (damage) of each $P' \in
\mathcal{P}'$ in the sums of Eq.~(\ref{eq:thirdMin}). In
particular, the following holds for any $\emptyset \subseteq
\mathcal{S} \subseteq \mathcal{P}$:
\begin{align*}
\min_{\mathcal{S}' \subseteq \mathcal{P}'} \vert U_{\mathcal{S}}
\triangle U_{\mathcal{S}'} \vert = 
\vert U_{\mathcal{S}} \cap P_j' \vert + \vert P_i' \vert - \vert U_{\mathcal{S}} \cap P_i' \vert \\ 
+ \nonumber \sum_{\substack{P' \in \mathcal{P}'\\P' \notin \{P_i',
    P_j'\}}} \min\{\vert U_{\mathcal{S}} \cap P'\vert, \vert P' \vert
- \vert U_{\mathcal{S}} \cap P' \vert\} \nonumber\\ 
= \vert U_{\mathcal{S}} \cap P_j' \vert + \vert P_i' \vert - \vert
U_{\mathcal{S}} \cap P_i' \vert + \nonumber\\
\sum_{\substack{P' \in \mathcal{P}'\\P' \notin \{P_i', P_j'\}}} \vert P' \vert
\min\{\frac{\vert U_{\mathcal{S}} \cap P'\vert}{\vert P' \vert}, 1 -
\frac{\vert U_{\mathcal{S}} \cap P' \vert}{\vert P' \vert}\}
\nonumber\\ 
= \vert U_{\mathcal{S}} \cap P_j' \vert + \vert P_i'
\vert - \vert U_{\mathcal{S}} \cap P_i' \vert +
\nonumber\\
\sum_{\substack{P' \in \mathcal{P}'\\P' \notin \{P_i',
    P_j'\}}} \vert P' \vert \peak(\frac{\vert U_{\mathcal{S}} \cap P'
  \vert}{\vert P' \vert}).
\end{align*}
\noindent Proposition~\ref{prop:constrained} below summarizes our
findings.

\begin{prop}
An optimal \cthree-correspondence
  $(\mathcal{S}, \mathcal{S}')$ under the constraint $P_i' \in
  \mathcal{S}'$ and $P_j' \in \mathcal{P}' \setminus \mathcal{S}'$ can
  be computed by first finding $\emptyset \neq \mathcal{S} \subsetneq
  \mathcal{S}$ that minimizes the term
\begin{align*}
\phistar_{i',j'}(\mathcal{S}) := \vert U_{\mathcal{S}} \cap P_j' \vert +
\vert P_i' \vert - \vert U_{\mathcal{S}} \cap P_i' \vert + \\ 
\sum_{\substack{P' \in \mathcal{P}'\\P' \notin \{P_i', P_j'\}}} \vert
P' \vert \peak(\frac{\vert U_{\mathcal{S}} \cap P' \vert}{\vert P'
  \vert})
\end{align*}

\noindent and then setting $\mathcal{S}'$ as in
Eq.~(\ref{eq:partner1Copy}).
\label{prop:constrained}
\end{prop}

\begin{prop}$\phistar_{i',j'}(\cdot)$ is submodular and not
    symmetric.
\label{prop:submodular_constrained}
\end{prop}

\begin{proof}
Sums and positive multiples of submodular functions are
submodular~\cite{Schrijver2003a}. Thus, since $f(\mathcal{S}) := \vert
U_{\mathcal{S}} \cap P_j' \vert$ and $g(\mathcal{S}) := \vert P_i'
\vert - \vert U_{\mathcal{S}} \cap P_i' \vert$ are modular functions,
submodularity of $\phistar_{i',j'}(\cdot)$ follows from
$h(\mathcal{S}) := \peak(\frac{\vert U_{\mathcal{S}} \cap P'
  \vert}{\vert P' \vert})$ being submodular for all $P' \in
\mathcal{P}'$. The latter was shown in the proof of
Proposition~\ref{prop:submodular}. To see that
$\phistar_{i',j'}(\cdot)$ is not symmetric, first note that
$\phistar_{i',j'}(\mathcal{S}) = \phistar_{j',i'}(\mathcal{P \setminus
  \mathcal{S}})$. Thus, symmetry of $\phistar_{i',j'}(\cdot)$ would
imply $\phistar_{i',j'}(\mathcal{S}) = \phistar_{j',i'}(\mathcal{S})$
for all $\emptyset \neq \mathcal{S} \subsetneq \mathcal{P}$. For a
counterexample see Figure~\ref{fig:exampleG}.
\end{proof}

The function $\phistar(\cdot)$ is symmetric, see
Proposition~\ref{prop:submodular-nondeg}. Thus, we can compute a
minimum cut basis of $\mathcal{P}$ \wrt $\phistar(\cdot)$ by finding a
certain collection of $\vert \mathcal{P} \vert - 1$ minimum
$P_s$-$P_t$ cuts of $\mathcal{P}$~\cite{Cheng1992a}. Minimum
$P_s$-$P_t$ cuts are defined as in Definition~\ref{def:s-t-cuts} with
the exception that ``minimum'' now is \wrt $\phistar(\cdot)$. In
contrast to the $\vert \mathcal{P} \vert - 1$ cuts in the minimum
basis \wrt $\phi_{\mathcal{P}'}(\cdot)$, the cuts in the minimum basis
\wrt $\phistar(\cdot)$ are crossing cuts, in general.

\paragraph*{\cfour-correspondences.}
These correspondences raise two major difficulties. First, the sets
$\mathcal{S}$ in mutual correspondences $(\mathcal{S}, \mathcal{S}')$
do not form a lattice family~\cite{Goemans95a}. The latter is a family
$\mathcal{L}$ of subsets of a set $V$ such that $A, B \in \mathcal{L}$
implies $A \cap B, A \cup B \in \mathcal{L}$. For an example of mutual
correspondences causing a lattice family conflict see
Figure~\ref{fig:exampleG}. Second, it can occur that there are no
\cfour-correspondences at all, which raises a serious problem to any
B\&B algorithm for finding
\cfour-correspondences.


\subsection{Optimal \ctwo-correspondences and mutual correspondences.}
\label{subsec:app:mutual}
Propositions~\ref{prop:char_nontriv} and~\ref{prop:char_nondeg} below
tells us that an optimal \ctwo-correspondence or
\cthree-cor\-res\-pon\-dence is either mutual or simple. Here, ``simple''
means that Eq.~(\ref{eq:simple}) is fulfilled.

\begin{prop}
If an optimal \ctwo-correspondence $(\mathcal{S}, \mathcal{S}')$
is not mutual, then
\begin{equation}
\label{eq:simple}
\vert \mathcal{S} \vert \in \{1, \vert \mathcal{P} \vert - 1\} \vee \vert \mathcal{S}' \vert \in \{1, \vert \mathcal{P}' \vert - 1\}.
\end{equation}

\label{prop:char_nontriv}
\end{prop}
\begin{proof}
Let $(\mathcal{S}, \mathcal{S}')$ be an optimal \ctwo-correspondence
that is not mutual. First assume that $(\mathcal{S}, \mathcal{S}')$
does not fulfill item $1.$ in Definition~\ref{def:mutual}. Then there
exists $P \in \mathcal{S}$ with $\vert U_{\mathcal{S} \setminus \{P\}}
\triangle U_{\mathcal{S}'} \vert < \vert U_{\mathcal{S}} \triangle
U_{\mathcal{S}'} \vert$. Since $(\mathcal{S}, \mathcal{S}')$ is an
optimal \ctwo-correspondence, the correspondence $(\mathcal{S}
\setminus \{P\}, \mathcal{S}')$ cannot fulfill \ctwo, \ie $\vert
\mathcal{S} \vert = 1$. Likewise, items $2.$, $3.$ and $4.$ imply
$\vert \mathcal{S} \vert = \vert \mathcal{P} \vert - 1$, $\vert
\mathcal{S}' \vert = 1$ and $\vert \mathcal{S}' \vert = \vert
\mathcal{P}' \vert - 1$, respectively.
\end{proof}

\noindent An analogous proof leads to an analogous characterization
of \cthree-cor\-res\-pond\-en\-ces.


\begin{prop}
If an optimal \cthree-cor\-res\-pon\-dence $(\mathcal{S}, \mathcal{S}')$ is
not mutual, then Eq.~(\ref{eq:simple}) holds.
\label{prop:char_nondeg}
\end{prop}

\subsection{Details on B\&B algorithm.}
\label{subsec:app:bb}
The following notation will make it easier to formulate our
B\&B algorithm presented as Algorithm~\ref{algo:BB}.
\begin{notation}
\label{notation:indices}
\wwlog the parts in $\mathcal{S}_s \cup \mathcal{S}_t \setminus
\{P_s, P_t\}$ are denoted by $\hat{P}_1, \dots, \hat{P}_{\vert
  \mathcal{S}_s \cup \mathcal{S}_t \vert - 2}$, and the indices
reflect the order in which the parts were added to $\mathcal{S}_s
\setminus \{P_s\}$ or $\mathcal{S}_t \setminus \{P_t\}$ (the larger
an index, the later the part was added).
\end{notation}

\begin{algorithm}[tbh]
  \caption{B\&B algorithm for finding a minimum
    $P_s$-$P_t$ cut.}
  \label{algo:BB}
  \begin{algorithmic}[1]
    \State $\mathcal{S}_s \gets \{P_s\}$, $\mathcal{S}_t \gets
    \{P_t\}$
    \State $bestSoFar \gets \infty$
    \Do
    \State $\greed(\mathcal{S}_s, \mathcal{S}_t, bestSoFar)$
    \If{$\mathcal{S}_s \cup \mathcal{S}_t = \mathcal{P}$}
    \Comment{\ie we have found a $P_s$-$P_t$ cut}
    \If{$b(\mathcal{S}_s, \mathcal{S}_t) < bestSoFar$}
    \Comment{$b(\mathcal{S}_s, \mathcal{S}_t) = \phi_{\mathcal{P}'}(\mathcal{S}_s) = \phi_{\mathcal{P}'} (\mathcal{S}_t)$}
    \State $\mathcal{S} \gets \mathcal{S}_s$
    \State $bestSoFar \gets b(\mathcal{S}_s, \mathcal{S}_t)$
    \EndIf
    \EndIf
    \State $i \gets 0$ \Comment{Beginning of undo}
    \Do
    \If{$\hat{P}_{\vert \mathcal{S}_s \cup \mathcal{S}_t \vert-2-i}
      \in \mathcal{S}_s$}
    \State $\mathcal{S}_s \gets \mathcal{S}_s \setminus \hat{P}_{\vert \mathcal{S}_s \cup \mathcal{S}_t \vert-2-i}$
    \Else \Comment{\ie $\hat{P}_{\vert \mathcal{S}_s \cup \mathcal{S}_t \vert-2-i}
      \in \mathcal{S}_t$}
    \State $\mathcal{S}_t \gets \mathcal{S}_t \setminus \hat{P}_{\vert \mathcal{S}_s \cup \mathcal{S}_t \vert-2-i}$
    \EndIf
    \State $i \gets i + 1$
    \doWhile{$(\mathcal{S}_s, \mathcal{S}_t) \neq (\{P_s\}, \{P_t\})
      \wedge \dejavu(A(\mathcal{S}_s, \mathcal{S}_t))$} \Comment{End
      of undo}
    \If{$(\mathcal{S}_s, \mathcal{S}_t) \neq (\{P_s\}, \{P_t\})$}
    \State $(\mathcal{S}_s, \mathcal{S}_t) \gets A(\mathcal{S}_s, \mathcal{S}_t)$
    \EndIf
    \doWhile{$(\mathcal{S}_s, \mathcal{S}_t) \neq (\{P_s\}, \{P_t\})$}
    \State {\bf return} $(\mathcal{S}, \mathcal{P} \setminus \mathcal{S})$
  \end{algorithmic}
\end{algorithm}

After initializing $\mathcal{S}_s$ and $\mathcal{S}_t$, our
B\&B algorithm calls $\greed(\mathcal{S}_s, \mathcal{S}_t,
\infty)$, see Algorithm~\ref{algo:BB} in Section~\ref{subsec:greedy}. In later calls of
$\greed(\mathcal{S}_s, \mathcal{S}_t, bestSoFar)$, we always have
$(\mathcal{S}_s \supsetneq \{P_s\} \vee \mathcal{S}_t \supsetneq
\{P_t\}) \wedge \mathcal{S}_s \cap \mathcal{S}_t = \emptyset$, and
$bestSoFar$ amounts to the minimum weight ($\phi_{\mathcal{P}'}$
value) of the $P_s$-$P_t$ cuts found so far (see lines 5-10 of
Algorithm~\ref{algo:BB}). 

The following definition will make it easier to address the remaining
questions whose answer was left open in Section~\ref{subsec:BB_basicAlgo}.

\begin{definition}
If $\hat{P}_{\vert \mathcal{S}_s \cup \mathcal{S}_t \vert - 2}$ is
contained in $\mathcal{S}_s$, the \emph{alternative} to $(\mathcal{S}_s,
\mathcal{S}_t)$ called $A(\mathcal{S}_s, \mathcal{S}_t)$
 is $(\mathcal{S}_s \setminus \{\hat{P}_{\vert
  \mathcal{S}_s \cup \mathcal{S}_t \vert - 2}\}, \mathcal{S}_t
\cup \{\hat{P}_{\vert \mathcal{S}_s \cup \mathcal{S}_t \vert -
  2}\})$. If $\hat{P}_{\vert \mathcal{S}_s \cup \mathcal{S}_t
  \vert - 2}$ is contained in $\mathcal{S}_t$, the \emph{alternative} to
$(\mathcal{S}_s, \mathcal{S}_t)$ is $A(\mathcal{S}_s, \mathcal{S}_t) :=
(\mathcal{S}_s \cup
\{\hat{P}_{\vert \mathcal{S}_s \cup \mathcal{S}_t \vert - 2}\},
\mathcal{S}_t \setminus \{\hat{P}_{\vert \mathcal{S}_s \cup
  \mathcal{S}_t \vert - 2}\})$.
\end{definition}

The answer to 2a) now is ``Undo the assignment of $\hat{P}_{\vert
  \mathcal{S}_s \cup \mathcal{S}_t \vert - 2}$. Keep undoing the
latest assignments until some $\hat{P}_{\vert \mathcal{S}_s \cup
  \mathcal{S}_t \vert - 2 - i}$, $i \geq 1$, is reached such that
$\greed(\cdot, \cdot, \cdot)$ has \emph{not} yet been called with the
first two arguments given by $A(\mathcal{S}_s, \mathcal{S}_t)$.'' In
the pseudocode of Algorithm~\ref{algo:BB}, a boolean function called
$\dejavu(\cdot, \cdot)$ is used to express whether $A(\mathcal{S}_s,
\mathcal{S}_t)$ has entered the call of $\greed(\cdot, \cdot, \cdot)$
before, see line 19 of Algorithm~\ref{algo:BB}. This line guarantees
termination of our B\&B algorithmB\&B algorithm. The answer to 2b) then
is ``call $\greed(\cdot, \cdot, \cdot)$ with $A(\mathcal{S}_s,
\mathcal{S}_t)$ and the current value of $bestSoFar$'' (see lines 21
and 4 of Algorithm~\ref{algo:BB}).

\if #0
\noindent We conclude this section with two implementation details.

\begin{itemize}
\item We implement the boolean function $\dejavu(\cdot, \cdot)$ as a
  boolean vector called $doneWith[\cdot]$. The two arguments
  $\mathcal{S}_s$, $\mathcal{S}_t$ of $\dejavu(\cdot, \cdot)$
  correspond to a single index $i$ of $doneWith$. Specifically, $i$ is
  the cardinality of $\mathcal{S}_s \cup \mathcal{S}_t \setminus
  \{P_s, P_t\}$ or, equivalently, the highest index, $\vert
  \mathcal{S}_s \cup \mathcal{S}_t \vert - 2$, in
  Notation~\ref{notation:indices}. The vector $doneWith$ has length
  $\vert \mathcal{P} \vert - 2$, and its entries are initialized to
  false. Anytime $\greed(\cdot, \cdot, \cdot)$ is called with
  $(\mathcal{S}_s, \mathcal{S}_t) = A(\hat{\mathcal{S}}_s,
  \hat{\mathcal{S}}_t)$ for some $\hat{\mathcal{S}}_s,
  \hat{\mathcal{S}}_t \subset \mathcal{P}$, $doneWith[i]$ is set to
  true. Furthermore, if backtracking goes back behind an
  index $j$, $doneWith[j]$ is set to false.
\item For efficient updates of $b(\mathcal{S}_s, \mathcal{S}_t)$ in
  Algorithm~\ref{algo:BB}, \ie when $\mathcal{S}_s$ or $\mathcal{S}_t$
  grows or shrinks by one part, we use distributions as in
  Section~\ref{subsec:goodC2_running_times} --- this time the \emph{rows} of the
  contingency table of $\mathcal{P}$ and $\mathcal{P}'$. Specifically,
  let $P \in \mathcal{P}$. Then $d_{P}[\cdot]$ is a vector of length
  $\vert \mathcal{P}' \vert$ defined by $d_{P}[j] := \vert P \cap P'_j
  \vert \mbox{~for~} 1 \leq j \leq \vert \mathcal{P}' \vert$. We
  define $d_{\mathcal{S}_s} := \sum_{P \in \mathcal{S}_s} d_{P}$. Now,
  adding or removing a part $P$ from $\mathcal{S}_s$ corresponds to
  the command $d_{\mathcal{S}_s} \gets d_{\mathcal{S}_s} + d_{P}$ and
  $d_{\mathcal{S}_s} \gets d_{\mathcal{S}_s} - d_{P}$,
  respectively. By carrying along $d_{\mathcal{S}_s}$ and
  $d_{\mathcal{S}_t}$, we can compute the bounds $b(\cdot, \cdot)$ as
  follows.

  \begin{equation*}
    \label{eq:bound2}
    b(\mathcal{S}_s, \mathcal{S}_t) = \sum^{\vert \mathcal{P}' \vert}_{j=1}
    \min\{d_{\mathcal{S}_s}[j], d_{\mathcal{S}_t}[j]\}.
  \end{equation*} 

Since the distributions $d_{P'}\lbrack \cdot \rbrack$ are basically
long vectors, our implementation is well-suited for parallel
processing.
\end{itemize}

\fi

\subsection{Extensions of B\&B from \ctwo to \cthree and \cfour.}
\label{subsubsec:extensions}
The extension from \ctwo to \cthree needs two adaptations. First, an
early exit $(\mathcal{S}_s, \mathcal{S}_s')$ or $(\mathcal{S}_t,
\mathcal{S}_t')$ must fulfill $\mathcal{S}_s' \neq \emptyset$ and
$\mathcal{S}_t' \neq \emptyset$, respectively. Second, assume that our
B\&B algorithm has reached a point where all $P \in
\mathcal{P}$ have been assigned to the $s$-side or to the $t$-side. If
$\mathcal{S}'$ is still in $\{\emptyset, \mathcal{P}\}$, we modify it
such that it is not in $\{\emptyset, \mathcal{P}\}$ anymore and such
that the damage to $\phi_{\mathcal{P}'}(\cdot)$ is minimum.

If $\cfour$-correspondences are to be found, the search can be
interrupted whenever there exists $P_t \in \mathcal{S}_t$ such that
$\vert P_t \cap U_{\mathcal{S}_s'} \vert > \vert P_t \vert / 2$. A second
analogous criterion for interrupting the search arises from exchanging
the roles of $s$ and $t$. Moreover, early exits $(\mathcal{S}_s,
\mathcal{S}_s')$ [$(\mathcal{S}_t, \mathcal{S}_t')$] have to be
checked for mutuality of $\mathcal{S}_s$ and $\mathcal{S}_s'$
[$\mathcal{S}_t$ and $\mathcal{S}_t'$]. Analogously, at any point
where all $P \in \mathcal{P}$ have been assigned to the $s$-side or to
the $t$-side, the current correspondence $(\mathcal{S}, \mathcal{S}')$
must be checked for mutuality of $\mathcal{S}$ and $\mathcal{S}'$.

\subsection{Running times.}
\label{sec:app:run}
The detailed running times of the algorithms under consideration are given in 
Tables~\ref{tab:running_BB} and~\ref{tab:running_greed} below.
\begin{table}[h!]
\caption{Running times (in seconds) for calculating the $\vert
  \mathcal{P} \vert - 1$ best correspondences using the B\&B algorithm
  from Section~\ref{sec:BB}. Minima, mean values and maxima are over
  10 runs (the community detection algorithm is non-deterministic).}
\label{tab:running_BB}
\begin{center}
\begin{small}
\scalebox{0.78}{
\begin{tabular}{ l | l | r | r | r | r }
Graph ID & Name & Min & Mean & Max\\ \hline \hline
 1 & \textsc{p2p-Gnutella}          & {0.052}  & {0.060} & {0.070}\\\hline
 2 & \textsc{PGPgiantcompo}         & {0.256}  & {0.313} & {0.379}\\\hline
 3 & \textsc{email-EuAll}           & {0.255}  & {0.370} & {0.574}\\\hline
 4 & \textsc{as-22july06}           & {0.292}  & {0.329} & {0.386}\\\hline
 5 & \textsc{soc-Slashdot0902}      & {1.048}  & {1.557} & {2.855}\\\hline
 6 & \textsc{loc-brightkite\_edges} & {4.309}  & {5.520} & {11.210}\\\hline
 7 & \textsc{loc-gowalla\_edges}    & {28.013}  & {50.265} & {240.330}\\\hline
 8 & \textsc{coAuthorsCiteseer}     & {19.871}  & {29.574} & {56.117}\\\hline
 9 & \textsc{wiki-Talk}             & {38.403}  & {1230.200} & {9554.700}\\\hline
 10 & \textsc{citationCiteseer}     & {12.297}  & {13.353} & {15.488}\\\hline
 11 & \textsc{coAuthorsDBLP}        & {27.791}  & {23.892} & {26.606}\\\hline
 12 & \textsc{web-Google}           & {20.414}  & {22.309} & {25.432}\\\hline
 13 & \textsc{coPapersCiteseer}     & {77.912}  & {356.130} & {1961.300}\\\hline
 14 & \textsc{coPapersDBLP}         & {38.824}  & {36.438} & {39.386}\\\hline
\end{tabular}}
\end{small}
\end{center}
\end{table}
\begin{table}[h!]
\caption{Running times (in seconds) for calculating $\vert \mathcal{P}
  \vert - 1$ correspondences using the algorithm \greed from
  Section~\ref{sec:BB}. Minima, mean values and maxima are over 10
  runs (the community detection algorithm is non-deterministic).}
\label{tab:running_greed}
\begin{center}
\begin{small}
\scalebox{0.78}{
\begin{tabular}{ l | l | r | r | r | r }
Graph ID & Name & Min & Mean & Max\\ \hline \hline
 1 & \textsc{p2p-Gnutella}          & {0.034} & {0.037} & {0.041}\\\hline
 2 & \textsc{PGPgiantcompo}         & {0.119} & {0.136} & {0.150}\\\hline
 3 & \textsc{email-EuAll}           & {0.134} & {0.166} & {0.205}\\\hline
 4 & \textsc{as-22july06}           & {0.125} & {0.148} & {0.189}\\\hline
 5 & \textsc{soc-Slashdot0902}      & {0.435} & {0.691} & {0.958}\\\hline
 6 & \textsc{loc-brightkite\_edges} & {2.102} & {2.173} & {2.304}\\\hline
 7 & \textsc{loc-gowalla\_edges}    & {12.980} & {14.411} & {15.479}\\\hline
 8 & \textsc{coAuthorsCiteseer}     & {7.852} & {8.340} & {9.086}\\\hline
 9 & \textsc{wiki-Talk}             & {23.290} & {28.893} & {36.077}\\\hline
 10 & \textsc{citationCiteseer}     & {6.198} & {6.801} & {7.144}\\\hline
 11 & \textsc{coAuthorsDBLP}        & {10.620} & {11.194} & {11.683}\\\hline
 12 & \textsc{web-Google}           & {8.815} & {9.575} & {10.243}\\\hline
 13 & \textsc{coPapersCiteseer}     & {18.849} & {20.035} & {22.214}\\\hline
 14 & \textsc{coPapersDBLP}         & {15.583} & {16.974} & {19.239}\\\hline
 15 & \textsc{as-skitter}           & {21.423} & {22.466} & {24.134}\\\hline
\end{tabular}}
\end{small}
\end{center}
\end{table}

\end{document}